\documentclass[acmsmall]{custom-arxiv}  

\usepackage{booktabs}

\usepackage{flushend}

\usepackage{amsmath,amsfonts}
\usepackage{amsthm}
\usepackage[english]{babel}
\usepackage{algorithm2e}
\usepackage{graphicx}
\usepackage{textcomp}
\usepackage{xcolor}

\usepackage{rotating}
\usepackage{tabularx}
\usepackage{pgfplots}

\newcommand{\featureset}{\ensuremath{I}}

\begin{document}

\title[Subjective Fairness in Privacy-Restricted Decentralised Conflict Resolution]{Agree to Disagree: Subjective Fairness in Privacy-Restricted Decentralised Conflict Resolution}

\author{Alex Raymond}
\email{alex.raymond@cl.cam.ac.uk}
\affiliation{%
 \institution{Department of Computer Science and Technology, University of Cambridge, UK \\ alex.raymond@cl.cam.ac.uk}
}

\author{Matthew Malencia}
\affiliation{%
 \institution{GRASP Laboratory, University of Pennsylvania, USA \\ malencia@seas.upenn.edu}
}

\author{Guilherme Paulino-Passos}
\affiliation{%
 \institution{Department of Computing, Imperial College London, UK \\ g.passos18@imperial.ac.uk}
}

\author{Amanda Prorok}
\affiliation{%
 \institution{Department of Computer Science and Technology, University of Cambridge, UK \\ asp45@cam.ac.uk}
}

\begin{abstract}

Fairness is commonly seen as a property of the global outcome of a system and assumes centralisation and complete knowledge. However, in real decentralised applications, agents only have partial observation capabilities. Under limited information, agents rely on communication to divulge some of their private (and unobservable) information to others. When an agent deliberates to resolve conflicts, limited knowledge may cause its perspective of a correct outcome to differ from the actual outcome of the conflict resolution. This is \textit{subjective unfairness}.

As human systems and societies are organised by rules and norms, hybrid human-agent and multi-agent environments of the future will require agents to resolve conflicts in a decentralised and \textit{rule-aware} way. Prior work achieves such decentralised, rule-aware conflict resolution through cultures: explainable architectures that embed human regulations and norms via argumentation frameworks with verification mechanisms. However, this prior work requires agents to have full state knowledge of each other, whereas many distributed applications in practice admit partial observation capabilities, which may require agents to communicate and carefully opt to release information if privacy constraints apply.

To enable decentralised, fairness-aware conflict resolution under privacy constraints, we have two contributions: (1) a novel interaction approach and (2) a formalism of the relationship between privacy and fairness. Our proposed interaction approach is an architecture for privacy-aware explainable conflict resolution where agents engage in a dialogue of hypotheses and facts. To measure the privacy-fairness relationship, we define subjective and objective fairness on both the local and global scope and formalise the impact of partial observability due to privacy in these different notions of fairness.

We first study our proposed architecture and the privacy-fairness relationship in the abstract, testing different argumentation strategies on a large number of randomised cultures. We empirically demonstrate the trade-off between privacy, objective fairness, and subjective fairness and show that better strategies can mitigate the effects of privacy in distributed systems. In addition to this analysis across a broad set of randomised abstract cultures, we analyse a case study for a specific scenario: we instantiate our architecture in a multi-agent simulation of prioritised rule-aware collision avoidance with limited information disclosure.
\end{abstract}


\maketitle

\begin{abstract}

\section{}
Fairness is commonly seen as a property of the global outcome of a system and assumes centralisation and complete knowledge. However, in real decentralised applications, agents only have partial observation capabilities. Under limited information, agents rely on communication to divulge some of their private (and unobservable) information to others. When an agent deliberates to resolve conflicts, limited knowledge may cause its perspective of a correct outcome to differ from the actual outcome of the conflict resolution. This is \textit{subjective unfairness}.

As human systems and societies are organised by rules and norms, hybrid human-agent and multi-agent environments of the future will require agents to resolve conflicts in a decentralised and \textit{rule-aware} way. Prior work achieves such decentralised, rule-aware conflict resolution through cultures: explainable architectures that embed human regulations and norms via argumentation frameworks with verification mechanisms. However, this prior work requires agents to have full state knowledge of each other, whereas many distributed applications in practice admit partial observation capabilities, which may require agents to communicate and carefully opt to release information if privacy constraints apply. 

To enable decentralised, fairness-aware conflict resolution under privacy constraints, we have two contributions: (1) a novel interaction approach and (2) a formalism of the relationship between privacy and fairness. Our proposed interaction approach is an architecture for privacy-aware explainable conflict resolution where agents engage in a dialogue of hypotheses and facts. To measure the privacy-fairness relationship, we define subjective and objective fairness on both the local and global scope and formalise the impact of partial observability due to privacy in these different notions of fairness.

We first study our proposed architecture and the privacy-fairness relationship in the abstract, testing different argumentation strategies on a large number of randomised cultures. We empirically demonstrate the trade-off between privacy, objective fairness, and subjective fairness and show that better strategies can mitigate the effects of privacy in distributed systems. In addition to this analysis across a broad set of randomised abstract cultures, we analyse a case study for a specific scenario: we instantiate our architecture in a multi-agent simulation of prioritised rule-aware collision avoidance with limited information disclosure.

\tiny
 \keyFont{ \section{Keywords:} fairness, privacy, multi-agent, dialogues, argumentation, explanations} %
\end{abstract}

\section{Introduction and Motivation}

Cognition and autonomy allow intelligent agents in nature to capture information about their surroundings and independently choose a course of action in consonance with their unique and individual decision-making process. Societies thus emerged as collectives of individuals with (mostly) collaborative intent, aided by implicit and explicit systems of norms and rules. The particular complexity of some of these collectives lead some observers to personify or anthropomorphise these groups, attributing a misleading notion of centralised intent and agency to a coherent, but fully decentralised system. However, even the most harmonious and conformable populations in reality exhibit differences in perspective and disagreements across their members.

In view of the above, we understand that humans do not make decisions with \textit{truly} global information by virtue of the implausible assumption of omniscience. Rather, each individual acts on their own \textit{subjective} perspectives of local and global status. This subjectivity is not a flaw but instead a fundamental truth of decentralised systems where information is ultimately incomplete or imperfect. In such systems, agents judge outcomes of conflicts based on their partial knowledge of the world, which can lead to perceptions of \textit{unfairness} when outcomes differ from what other peers perceive as correct. Moreover, individual perspectives can differ drastically when agents \textit{choose to} retain information under concerns of privacy. As such, it is germane to transpose such considerations to human-agent and multi-agent systems of the future.

The concepts of fairness and justice are prevalent issues in human history along millenia and have been a central topic in areas such as economics \citep{Marx1875CritiqueGotha}, philosophy \citep{Rawls1991JusticeMetaphysical}, medicine \citep{Moskop2007TriagePrinciples}, and more recently, computer science \citep{li_interdisciplinary_2017}. However, fairness is predominantly regarded as a \textit{global} property of a system \citep{selbst_fairness_2019, verma_fairness_2018, narayanan_translation_2018}. In allocation problems such as nationwide organ transplant allocation, fairness concerns the outcome, i.e., whether patients are prioritised to receive donor organs accurately and without discrimination \citep{bertsimas_fairness_2013}. Assumptions of global scope and complete knowledge come naturally, as we can only guarantee that a system is fair to all if the information regarding all subjects involved is known. Those assumptions form an \textit{objective} concept of fairness.

In this work, we elicit a provocation to the habitual definition of fairness: in decentralised applications, where the assumptions of complete global knowledge are withdrawn, can we also discuss fairness with regards to the individual perception of each agent in a system, i.e. representing a \textit{subjective} notion of fairness? If knowledge is incomplete, can we enable agents to understand that apparently (subjectively) negative decisions are globally and objectively fair, i.e., decisions actually made \textit{for the greater good?}

\textbf{Explanations, Dialogues, and Privacy.} As the gap between objective and subjective fairness resides on the need for reasoning and justification, \textit{explanations} \citep{Cyras2019ExplanationsDispute, Rosenfeld2019ExplainabilitySystems, Sovrano2021FromExplanation} artlessly lend themselves as desirable tools to address this issue. This is the focus of much of prior literature concerning explanations to human users. However, also a concern is about explanations between artificial agents themselves, or from a human to an artificial agent. Thus, in a fairness-aware decentralised system with disputed resources, agents with explainable capabilities can expound the reasons for deserving resources in a \textit{dialogue} that is informed by their respective observations and grounded to a mutually agreed-upon definition of what is and is not fair, or a \textit{fairness culture}.%

Since the concerns for fairness also gravitate around the integration of humans and agents together via reasoning, we lean on architectures for explainable human-agent deconfliction \citep{Raymond2020Culture-BasedDeconfliction}. Ideally, participants in a multi-agent system should engage in dialogues \citep{Jakobovits1999DialecticFrameworks, Modgil2009ArgumentationGoals} and provide explanations for claiming specific resources. One approach for this problem is computational argumentation, a model for reasoning which represents information as arguments and supports defeasible reasoning capabilities -- as well as providing convincing explanations to all agents involved \citep{Dung1995OnGames, Amgoud2009UsingDecisions, Fan2015OnArgumentation}.

However, not every scenario or application allows for unrestricted exchange of information between agents \citep{dwork_algorithmic_2014}. Sensitivity and confidentiality aggravate the poignancy of \textit{privacy} concerns, rendering the \textit{subjective} perception of fairness issue, at best, non-trivial when privacy is concerned. We defend \textit{subjective fairness} as a perennial dimension of the fairness argument whenever a system has non-global knowledge or \textit{privacy} is concerned.

Whilst the dimensions of privacy, subjective fairness, and objective fairness are desirable in all systems, the simultaneous maximisation of all those aspects is often irreconcilable. If the agents %
lack full information, there cannot be any guarantee of an objectively fair outcome. Additionally, if there are privacy considerations, subjective fairness cannot be guaranteed as the agent requesting a resource does not receive justification for denial. Therefore, with privacy, neither objective nor subjective fairness is achieved.

This perspective can be observed in society. We illustrate this point as a situation where two passengers are disputing a priority seat in public transport. Even if both individuals are truthful and subscribe to the same culture of fairness (e.g. both agree that a pregnant person should have priority over someone who just feels tired), they must still share information regarding their own personal condition (and therefore, their reasons to justify their belief) to determine who should sit down. How much is each person willing to divulge in order to obtain that resource? Assuming that these individuals have a finite and reasonably realistic threshold on how much privacy they are willing to abdicate for a seat, certain situations will inevitably engender the dilemma of either: \textit{a)} \textit{forfeiting the dispute and conceding the resource} to observe their privacy limit, or \textit{b)} \textit{exceeding their reservations in privacy} and revealing more information than they should in order to remain in contention. 

With adamantine restraints on privacy, contenders will \textit{always} concede and end the debate if presenting a superior reason would exceed their limits of divulged information. Note that this does not mean that the conceding party \textit{agrees} with the outcome. After all, if a debate was lost \textit{exclusively} due to the inviolability of privacy, then the loser must still hold an argument that they would have used to trump their opponent, had it not been precluded by privacy limits. \textit{Objectively}, the fair decision can only be guaranteed if all relevant information is made available by all parties for a transparent and reasonable judgement. In the \textit{subjective} perception of the bested agent, they still believe they have a superior reason at that point in time, and are left with no choice but to civilly \textit{agree to disagree}.

\textbf{Related Work.} Studies of privacy in multi-agent systems have gained recent popularity \citep{Such2014ASystems, Torreno2017CooperativeSurvey, Prorok2017Privacy-preservingSystems}. More closely, \citep{Gao2016Argumentation-basedPreserved} also propose the use of argumentation in privacy-constrained environments, although applied to distributed constraint satisfaction problems. Their approach, however, treats privacy in an absolute way, while in our notion is softer, with information having costs, and we consider varying degrees of privacy restrictions. 

Contemporaneously, the burgeoning research on fairness in multi-agent systems focuses on objective global fairness, assuming complete knowledge about all agents \citep{bin-obaid_fairness_2018}. Some works break the global assumption by applying fairness definitions to a neighbourhood rather than an entire population \citep{emelianov_price_2019} or by assuming that fairness solely depends on an individual \citep{nguyen_local_2016}. The former studies objective fairness of a neighbourhood, assuming full information of a subset of the population subset, whilst the latter assumes agents have no information outside of their own to make judgements about fairness. These works do not address fairness under partial observability, wherein agents have partial information on a subset of the population, which we call subjective local fairness.

To study privacy and subjective fairness in distributed multi-agent environments, we look to previous work in explainable human-agent deconfliction. The architecture proposed in \citep{Raymond2020Culture-BasedDeconfliction} introduces a model for explainable conflict resolution in multi-agent norm-aware environments by converting rules into arguments in a \textit{culture} (see Section~\ref{section:argumentation}). Disputes between agents are solved by a dialogue game, and the arguments uttered in the history of the exchange compose an explanation for the decision agreed upon by the agents.

Notwithstanding its abstract nature, this architecture relies on two important assumptions: \textit{i)} that agents have complete information about themselves and other agents; and \textit{ii)} that dialogues extend indefinitely until an agent is cornered out of arguments -- thus being \textit{convinced} to concede. In most real-life applications, however, those assumptions occur rather infrequently. Fully decentralised agents often rely on local observations and communication to compose partial representations of the world state, and indefinite dialogues are both practically and computationally restrictive. We build on the state of the art by extending this architecture to support gradual disclosure of information and privacy restrictions.

\textbf{Contributions.} This paper stages the ensuing contributions: 
\begin{itemize}
    \item We present an architecture for multi-agent privacy-aware conflict resolution using cultures and dialogue games.
    \item We introduce a formalisation of the subjective fairness problem and the corresponding privacy trade-off.
    \item We simulate interactions in random environments and compare how different argumentation and explanation strategies perform with regards to our fairness metrics.
    \item We instantiate a multi-agent prioritised collision avoidance scenario and demonstrate practical instances of subjective fairness and privacy limitations.
\end{itemize}

The structure of this paper follows as: Section~\ref{section:background} introduces prior background from literature used in our study. In Section~\ref{section:architecture}, we introduce an argumentation-based architecture for decentralised conflict resolution under privacy-restricted communication. We provide an abstract formalism of the dimensions of fairness in conflict resolution in Section~\ref{section:fairness}. Then, we use the contributions of those two sections in Sections~\ref{section:random-experiments} and~\ref{section:boats} to perform empirical experiments on abstract conflicts and on an applied multi-agent collision avoidance scenario, respectively. Our results show that using better argumentative strategies suggests Pareto improvements for the fairness-privacy trade-off.

\section{Background}
\label{section:background}

We introduce the required theoretical background from present literature used in the development of our study and methods. Section ~\ref{section:argumentation} abridges the concept of argumentation frameworks, dialogues, followed by cultures in Section~\ref{section:cultures}.

\subsection{Argumentation and Dialogues}
\label{section:argumentation}

The concept of \textit{culture} \citep{Raymond2020Culture-BasedDeconfliction} reflects a collective agreement of norms and priorities in a multi-agent system, checked dynamically by means of \textit{verifier functions}. Abstract Argumentation Frameworks~\citep{Dung1995OnGames} are used to underpin the mechanics of cultures and conflict resolution dialogues.

\begin{definition} [Argumentation Framework]
An \emph{argumentation framework} is a digraph $AF = (\mathcal{A}, \mathcal{R})$, where $\mathcal{A}$ is a set of arguments (vertices) and $\mathcal{R} \subseteq \mathcal{A} \times \mathcal{A}$ is a set of attack relations between arguments (arcs). We say $\textit{attacks}(a,b)$ holds iff $(a,b) \in \mathcal{R}$. Stating $\textit{attacks}(a,b)$ is intuitively equivalent to defining that an argument $b$ is defeated by argument $a$. Likewise, a set $S \subseteq \mathcal{A}$ of arguments attacks another set of arguments $T$ (or $T$ is attacked by $S$) if any argument in $S$ attacks an argument in $T$. An argument $a \in \mathcal{A}$ is \textit{acceptable} with respect to a set $S$ of arguments iff for each argument $b \in \mathcal{A}$ that attacks $a$ there is a $c \in S$ that attacks $b$. In that case, $c$ is said to \textit{defend} $a$. A set $S$ of arguments is \textit{conflict-free} iff for all $a, b \in S$, $(a,b) \not\in \mathcal{R}$, and \textit{admissible} iff it is conflict-free and all its arguments are acceptable with respect to $S$.
\end{definition}

The notion of \textit{extensions} \citep{Doutre2001PreferredComputation} presents itself as semantics of acceptance for sets of arguments.

\begin{definition} [Preferred Extension]
Let $AF = (\mathcal{A}, \mathcal{R})$ be an argumentation framework and $S \subseteq \mathcal{A}$ be a set of arguments. We say $S$ is a \textit{preferred extension} of $AF$ iff: $S$ is conflict-free, defends all arguments in $S$, and is a maximal element among all admissible sets, with respect to set-theoretical inclusion. If an argument $a \in \mathcal{A}$ is present in all preferred extensions of $AF$, we say $a$ is \textit{sceptically accepted} in the preferred extensions.
\end{definition}

With those definitions in place, agents may now partake in dialectical interaction to resolve conflicts. Those dialogues are sequential in nature, where agents utter arguments and attempt to defeat the previously-used argument by their adversary. When an agent can no longer choose a valid argument, they lose the game and concede the contended resource to the winner. \citet{Jakobovits1999DialecticFrameworks} delineate a formalism for dialectical games, and we adapt some of their definitions below for single-argument dialogues.

\begin{definition} [Dialogue]
\label{definition:dialogue}
Let a player $w \in \{pr, op\}$ and an argument $a \in \mathcal{A}$. The adversary of $pr$ is denoted $\overline{pr} = op$, and $\overline{op} = pr$. A \textit{move} is a pair $(w, a)$. For a move $m = (w, a)$, we use $\textit{player}(m)$ to denote $w$ and $\textit{arg}(m)$ to denote $a$. A move is considered \textit{legal} iff it does not attack itself and is not already attacked %
by a previously-uttered argument. We say a \textit{dialogue} $D$ is any countable sequence $m_0,m_1,\dots,m_n$ of moves that satisfies:
\begin{enumerate}
    \item $\textit{player}(m_{i+1}) = \overline{\textit{player}}(m_i)$, i.e. the players take turns.
    \item $m_{i+1}$ (i.e. the next move) is legal.
    \item $m_{i+1} \notin \{m_0, m_1, \dots, m_i\}$, i.e. a move cannot be repeated.
    \item $\textit{attacks}(\textit{arg}(m_{i+1}), \textit{arg}(m_i))$, it attacks the adversary's last move
    \item $\textit{player}(m_0) = pr$, i.e. the proponent makes the first move.
\end{enumerate}

The dialogue $D$ is said to be \textit{about the position} $\textit{arg}(m_0)$. 

If there is no dialogue $D' = m_0,m_1,\dots,m_n,\dots,m_{n'}$ that extends $D$, then we say that $\textit{player}(m_n)$ is the \textit{winner} of $D$.
\end{definition}

\subsection{Cultures}
\label{section:cultures}
In such frameworks, arguments and attacks are static by nature and assumed true. Informally, the surviving sets of arguments under certain semantics (extensions) represent coherent conclusions given a fixed argumentation framework. This inflexibility is a strong limitation for dynamic interactive systems, where agents can have varied states and contexts, causing arguments to lose validity depending on the circumstances. \textit{Cultures} for conflict resolution \citep{Raymond2020Culture-BasedDeconfliction} embody rules and norms as high-level, dynamic arguments that are shared and agreed upon between all agents in a system, and act as a template for generating argumentation frameworks dynamically to represent possible contexts and situations that arise in a multi-agent system. During inference, arguments are dynamically \textit{verified} for correctness given the required information from both agents and the context of a specific dispute.

\begin{definition} [Culture]
\label{definition:culture}
Let $A$ be the set of all existing agents in a multi-agent environment. Let any two players $pr, op \in A$ be the proponent and opponent in a dialogue game. We say a \textit{motion} is any argument $a \in \mathcal{A}$ that may be used by proponent $pr$ to request a contended resource from opponent $op$. Let $\mathcal{K} \subseteq \mathcal{A}$ be the set of all motions in $\mathcal{A}$. Let $\mathcal{R}$ be the set of attacks between arguments in $\mathcal{A}$. We say a system has a \textit{culture} $C = (\mathcal{A}, \mathcal{R}, \mathcal{K})$ iff $|\mathcal{K}| > 0 $ and all agents share $C$ as their culture.
\end{definition}

\begin{definition} [Argument Verification]
Let $a \in \mathcal{A}$ and $w \in \{pr,op\}$ be an argument and a player, respectively. We denote $a$ as \textit{demonstrable} by agent-player $w$ iff checking the correctness of that argument admits a finite and computable decision procedure. Let $\zeta$ denote the set of all possible contexts in the environment. $\forall a \in \mathcal{A}$, $a$ admits a predicate function $v_a: w, z \rightarrow \{\texttt{True}, \texttt{False}\}$, where $w$ represents the player and $z \in \zeta$ is a context. We say $v_a$ is the \textit{verifier} function of argument $a$. Motions are hypothetical and their verifier functions always return $\texttt{True}$. An argument $a$ is \textit{demonstrably true} by player $w$ iff $v_a(w, z) = \texttt{True}$.
\end{definition}

\section{Privacy-aware Cultures}
\label{section:architecture}

The cultures introduced in \citet{Raymond2020Culture-BasedDeconfliction} enable explainable conflict resolution in multi-agent norm-aware environments. However, these cultures require full state knowledge among agents -- and conflicts are resolved through indefinitely-long dialogues. In real-world distributed applications, agents only have partial observations of other agents, real world constraints prevent limitless dialogue, and privacy concerns govern the types and amount of information that agents divulge.

\begin{example}
\label{example:belle-charlotte}
Suppose 2 agents, Belle and Cadence, are disputing a priority seat on public transport. They share a culture $C = (\mathcal{A}, \mathcal{R}, \mathcal{K})$ where $\mathcal{A} = \{\gamma, a, b\}$ contains the arguments 
\begin{itemize}
    \item (motion) $\gamma \equiv$ \textit{`I think I should have this seat.'}
    \item $a \equiv$ \textit{`I am older than you.'}
    \item $b \equiv$ \textit{`I have a more serious health condition.'}
\end{itemize}
The relations $\mathcal{R} = \{(a, \gamma), (b, \gamma), (b, a)\}$ determine the priority of those concepts in the agents' shared culture. If Belle and Cadence do not know each other's age and health condition, they will have no information to verify arguments $a$ and $b$, i.e., they cannot argue that they are \textit{demonstrably} older or more infirm than their adversary. At this stage, any of those arguments could only be raised as a \textit{hypothesis}. Additionally, both Belle and Cadence would need to break privacy and mutually reveal some personal information to reach a decision.
\end{example}

To allow agents to engage in dialogues under partial information and privacy constraints, we present a novel explainable conflict resolution architecture. We begin by looking at the aspect of \textit{alteroceptive cultures}, our proposed mechanism organising the space of interactions for privacy-aware conflict resolution. `Alteroceptive' is our coinage from a \textit{portmanteau} of the Latin word \textit{alter} (other) + the word \textit{reception}. It shall connote a \textit{`sense of other.'}

\subsection{Alteroceptive Cultures}

As shown in Example~\ref{example:belle-charlotte}, when no information is available and arguments are comparative, agents can only raise \textit{hypotheses}. Those serve as media for sharing information between agents, by enforcing that every agent shares their corresponding piece of information when eliciting a hypothesis, in order to avoid infinite exchanges arising from cycles of empty suppositions.

Once an agent raises a hypothesis and shares their pertaining information regarding an argument, the adversary should have enough information to effectively \textit{verify} the argument into a fact, since it has full knowledge of its own description and was given the other agent's partial description. This can be done without necessarily sharing further information.

Additionally, the relations between arguments need to be woven in such a way as to guarantee that hypotheses and facts only defeat the arguments pertaining to the adversary. In Example~\ref{example:belle-charlotte}, if the same agent wants (and is able) to utter both arguments $a$ and $b$, the attack $(b, a)$ in the culture should not yield an actual attack in an instantiation if both arguments are articulated by the same agent.
After all, having a trumping %
condition would be yet another reason in favour of -- not against -- the agent's claim. Hence, assuring that attacks only occur between adversaries renders a bipartition in the culture and its arguments. The conjunction of those elements culminates in an \textit{alteroceptive culture}.

\begin{definition} [Alteroceptive Culture]
\label{definition:alteroceptive}
Let $C = (\mathcal{A}, \mathcal{R}, \mathcal{K})$ be a culture%
. We define the \textit{expansion function} $\kappa$, that maps an argument $a \in \mathcal{A}$ to a set of new arguments (not in $\mathcal{A}$). For every $a \in \mathcal{A}\setminus\mathcal{K}$, $\kappa(a) = \{a_{H}^{pr}, a_{H}^{op}, a_{F}^{pr}, a_{F}^{op}\}$ four new arguments that represent, respectively:
\begin{itemize}
    \item $a_{H}^{pr}$: hypothesis ($H$) of $a$ where proponent $pr$ wins,
    \item $a_{H}^{op}$: hypothesis ($H$) of $a$ where opponent $op$ wins,
    \item $a_{F}^{pr}$: verified-fact ($F$) of $a$ where proponent $pr$ wins,
    \item $a_{F}^{op}$: verified-fact ($F$) of $a$ where opponent $op$ wins.
\end{itemize}
For motions in $\gamma \in \mathcal{K}$, $\kappa(\gamma) = \{\gamma_{H}^{pr}, \gamma_{H}^{op}\}$, since motions do not admit verification. We define the \textit{expanded set of arguments (in an alteroceptive culture)} $\mathcal{A}_x = \bigcup\limits_{a \in \mathcal{A}} \kappa(a)$. The \textit{expanded set of attacks} $\mathcal{R}_x$ is constructed from $\mathcal{R}$ and $\mathcal{A}$, and satisfies the following rules (see Figure~\ref{fig:alteroceptive}). For each element of $\kappa(a)$, where $a, b \in \mathcal{A}$ and $(a, b) \in \mathcal{R}$. For all $w \in \{pr, op\}$:
\begin{enumerate}
    \item $(a_{H}^{w}, a_{H}^{\overline{w}}) \in \mathcal{R}_x$ if $a \not \in \mathcal{K}$, i.e., non-motion hypotheses mutually attack each other;
    \item $(a_{F}^{w}, a_{F}^{\overline{w}}) \in \mathcal{R}_x$, i.e., verified-facts mutually attack each other; 
    \item $(a_{F}^{w}, a_{H}^{\overline{w}}) \in \mathcal{R}_x$  i.e., each verified-fact attacks their adversary's hypothesis;
    \item $(a_{H}^{w}, b_{H}^{\overline{w}}) \in \mathcal{R}_x$ and $(a_{H}^{w}, b_{F}^{\overline{w}}) \in \mathcal{R}_x$, i.e., hypotheses reproduce their original attacks to both hypotheses and verified-facts;  %
    \item there are no more elements in $\mathcal{R}_x$.
    \end{enumerate}

We say $C_x(C) = (\mathcal{A}_x, \mathcal{R}_x)$ is an \textit{alteroceptive expansion} of $C$. %
\end{definition}

\begin{figure}
\includegraphics[width=\linewidth]{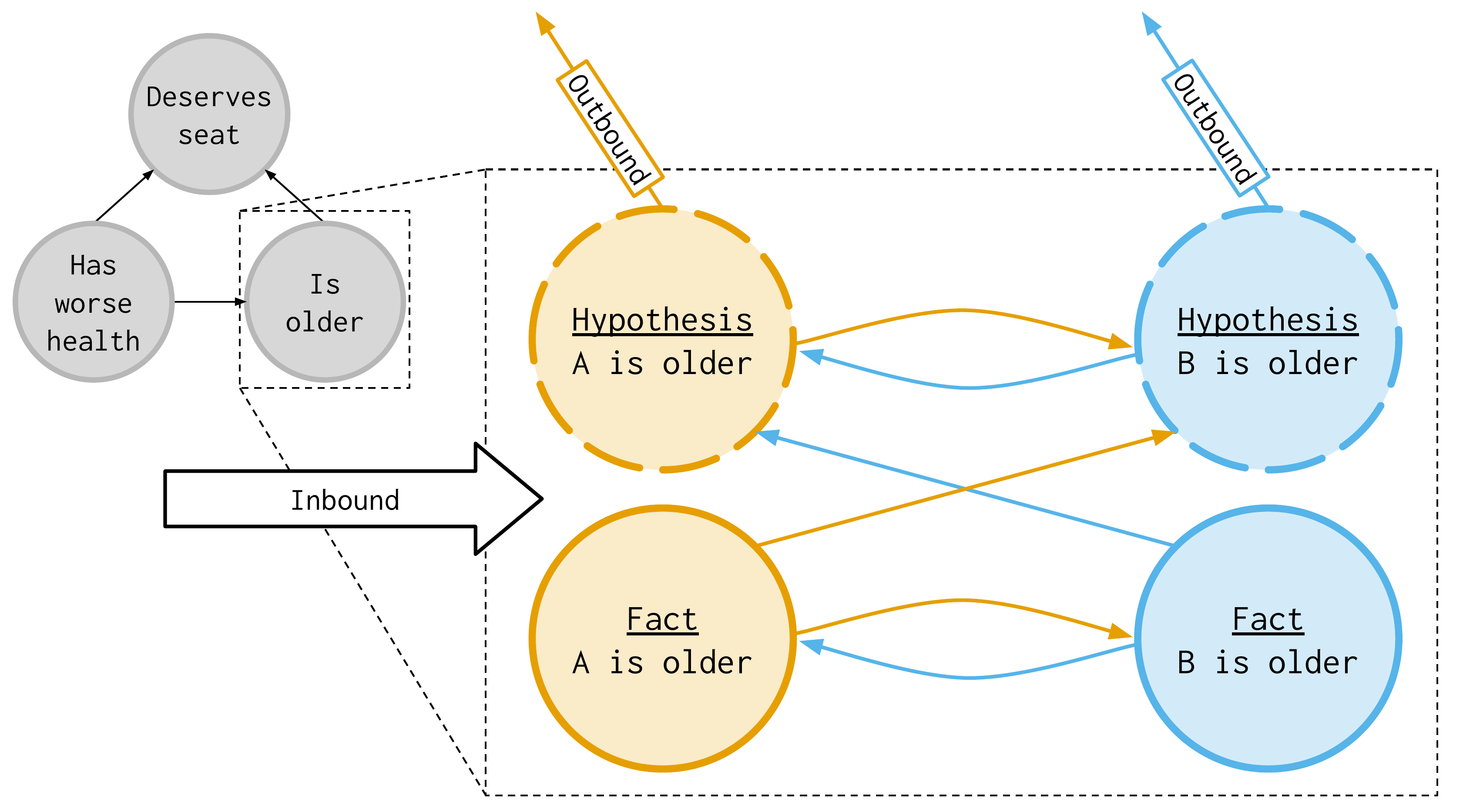}
\caption{Alteroceptive cultures transform the arguments of a given culture (grey) into  hypotheses (dashed nodes) and verified-facts (solid nodes). Shown here from Example~\ref{example:belle-charlotte}, argument $a$, "Is older", is transformed into the respective nodes for the orange opponent, $a_{H}^{op}$ and $a_{F}^{or}$, and the blue proponent, $a_{H}^{pr}$ and $a_{F}^{pr}$.
The inbound attack (left) from "Has worse health" attacks all four nodes and the outbound attacks (top) attack "Deserves seat" from the hypotheses.}
\label{fig:alteroceptive}
\end{figure}

The separation between hypotheses and facts is important to introduce a mechanism for gradual disclosure of information.

\begin{example} [cont.~\ref{example:belle-charlotte}]
Given the original culture seen in Example~\ref{example:belle-charlotte}, we can generate an extended set of arguments $\mathcal{A}_x = \bigcup \{\kappa(\gamma), \kappa(a), \kappa(b)\}$ and the corresponding $\mathcal{R}_x$ as described in Def.~\ref{definition:alteroceptive}. Suppose, then, that Cadence raises the motion $\gamma_{H}^{\textit{Cad}} \equiv $\textit{`I think I should have this seat,'} to which Belle promptly replies with $a_{H}^{\textit{Belle}} \equiv $ \textit{`I may be older than you, I am 60 years old.'} Figure~\ref{fig:alteroceptive} illustrates the expanded argument set in question. Note that Belle had to break privacy and reveal their age in order to use this argument. Cadence now has the necessary information (since they know their own age) to verify argument $a_{F}^{\textit{Cad}}$. If they succeed in verifying this argument (i.e., they are older than Belle), they have two options for rebuttal, depending on what information they intend to share:
\begin{enumerate}
    \item Use $a_{F}^{\textit{Cad}}$, reveal their age and refute Belle's claim of being older;
    \item or use $b_{H}^{\textit{Cad}}$, ignore the age dispute and move on to the health argument, revealing their health condition and hypothesising that they might be more ill than Belle.
\end{enumerate}
\end{example}

The example above shows that, for her next move, Cadence can choose to either reveal her age or health condition to progress the debate. This decision is affected by which aspect of its description an agent would like to keep private. Laying out arguments within this structure \textbf{still allows for the same dialogue game rules as before, but without the original assumption of complete information.} A dialogue under this framework would carry out normally and extend indefinitely as agents exchange moves until one of the agents loses by running out of arguments. Information can, therefore, still be communicated unreservedly. In the interest of quantifying relinquished privacy, we associate arguments to \textit{privacy costs}.

\subsection{Privacy-Aware Dialogues}
Every argument in an alteroceptive culture corresponds to a concession of privacy through communication of one's partial description. Fortunately, those changes preserve the structure of an argumentation framework, which makes alteroceptive cultures compatible with prior mechanics of dialogue, as well as extensions. Analogously to how humans are comfortable with sharing some types of personal information but not others, some features in agents' descriptions might be considered more sensitive, and thus having a higher \textit{privacy cost} to reveal. We combine the rules of dialogue (Def.~\ref{definition:dialogue}) and the notion of \textit{privacy cost} to provide an instantiation for our architecture. 

\begin{definition} [Privacy Cost and Budget]
Let $\mathcal{A}_x$ be the extended set of arguments of an alteroceptive culture. We say $\tau: \mathcal{A}_x \rightarrow \mathbb{Z}^+$ is a \textit{privacy cost} function. Let $A$ be the set of all agents. We define $\beta: A \rightarrow \mathbb{Z}^+$ as the \textit{privacy budget} function of an agent.
\end{definition}

\begin{definition} [Privacy-Aware Dialogue]
Let $w \in \{pr, op\}$ be any two players and $D = m_0,m_1,\dots,m_n$ be a dialogue. Let $moves(w, n) = \bigcup\limits_{m_i \in D}\{m_i \mid \textit{player}(m_i) = w \text{ and } i \leq n\}$ denote all the moves of a player $w$ up to round $n$. We say a dialogue $D$ is \textit{privacy-aware} iff, in addition to the criteria in Def.~\ref{definition:dialogue}, it also satisfies, for all $n$:
\begin{enumerate}
\setcounter{enumi}{5}
    \item $\beta(\textit{player}(m_{i+1})) \geq \sum\limits_{m_i \in moves(w, n)} \tau(\textit{arg}(m_i)),$ \\i.e., the player cannot cumulatively spend more than its privacy budget.
\end{enumerate}

\end{definition}

Our definition imposes a hard limit privacy: during a dialogue, agents cannot use an argument that would aggregately exceed their pre-defined privacy budgets. With finite privacy budgets, \textbf{the assumption of dialogues extending indefinitely until reaching a unanimous conclusion no longer holds.} Disputes end with an agent losing in one of two different ways. Either: \textit{i)} the agent runs out of arguments and is \textit{convinced} out of the dialogue, or \textit{ii)} the agent still has valid hypotheses and/or verified-facts that could attack the last move, but cannot \textit{afford} to use them and is forced to concede. In the former case, both agents agree on the correct outcome as the losing agent did not have any valid reasons to challenge the winner. In the latter case, however, the agent concedes and interrupts the dialogue before having all arguments successfully refuted - as it prefers to preserve their privacy. The agents tolerate the result but do not agree with it: they \textit{agree to disagree}. 

Notably, the choice of arguments can have a decisive impact on the dialogue game. If we consider scenario \textit{(ii)} above to be less desirable than \textit{(i)} for both agents, both agents must maximise the effectiveness of their moves so that dialogues end on the basis of argumentation instead of privacy budgets. Yet, optimal strategies are not readily available, as partial observability qualifies our privacy-aware dialogue game as a game of incomplete information.\footnote{Typically such games are analysed by a completion approach using Bayesian reasoning, and thus turned into imperfect information games \citep{Binmore1992FunGames}. The alternative allows agents to adopt mixed strategies, using information about an underlying probability distribution. We consider only pure strategies due to fewer needed assumptions, where agents do not use any probabilistic information. However, for pure strategies the optimal solution can be intractable \citep{Blair1993GamesInformation,Reif1984TheInformation}, and we thus limit ourselves to simple ones in the next sections.} 

With that in place, we next introduce our foundational terminology and abstract definitions governing the problem of subjective fairness in privacy-restricted decentralised conflict resolution.

\section{Problem Definition}
\label{section:fairness}

While the new argumentation architecture presented above enables dialogues under the partial information of privacy constraints, such partial observability may cause \textit{discrepancies between the outcome of a conflict resolution and the agents’ individual perception} of what the correct outcome should be, which we define as subjective unfairness. This new slant on the perspective of fairness introduces another dimension for conflict resolution whenever a system has privacy reservations or non-global knowledge. This section thus defines a foundational mathematical formalism to study fairness and its relationship to privacy in decentralised conflict resolution.\footnote{We dissociate this formalism from Section~\ref{section:architecture}'s emphasis on argumentation frameworks to the higher-level space of conflict resolution, as AFs are one of many conflict resolution methods. We hope others in the conflict resolution domain will build on this work and study subjective fairness using their tools of choice.}

Consider a situation in which a population of agents in a multi-agent system interacts and disputes advantages or resources. Agents interact in a pairwise manner via \textit{disputes}, and a final system state is achieved through a sequence of such interactions. We observe whether this final state is fair given complete knowledge: this we call \textit{objective fairness}. A second aspect is whether the outcome is fair from the perspective of each agent. This is \textit{subjective fairness}. Scope further augments our setting: the term \textit{global} relates to population-wide observations, whilst \textit{local} concerns pairwise interactions.

We remove the assumptions of unrestricted information-sharing, instead assuming that \textit{privacy} restricts an agent's willingness to divulge information. Particularly, we analyse how privacy between agents is an impediment towards a guarantee of both objectively and subjectively fair outcomes. We formalise such concepts below.

\subsection{Information, Privacy and Fairness}

We postulate that agents possess idiosyncratic features with varied values across a population, and that the set of all features that describe
an agent is, by reason, called a \textit{description}. Assume a distributed multi-agent setting where agents have perfect knowledge about their own features but none-to-partial knowledge about other agents' descriptions. We first define the (full) description of an agent as an $n$-tuple of features.%

\begin{definition}[Description]
For the rest of this section, we consider a fixed set $A$ of existing agents in an environment and an underlying set $F$ of all possible values of a feature. The set of \textit{features} is $\featureset = \{1, \dots, n\}$. We define $\mathcal{F} = F^n$ as the set of \textit{feature descriptions}, that is, the set of n-tuples, each entry representing the value of a feature. The \textit{description function} $d: A \rightarrow \mathcal{F}$ is a function that maps each agent to its full description in terms of features.
\end{definition}

Since we are considering privacy-sensitive applications, we attribute a \textit{privacy cost} to each feature, which is then aggregated to produce a cost for the full description. For brevity, all further mentions of \textit{cost} shall refer exclusively to \textit{privacy} cost. The constant \(\texttt{unknown}\) represents an undisclosed value of a feature.

\begin{definition}[Privacy Cost]
\label{definition:privacy-cost}
The set of \textit{private features} is $F_p = F \cup \{\texttt{unknown}\}$. A \textit{partial description} is an element of the set $\mathcal{F}_p = {F_{p}}^n$, that is, an n-tuple with each entry corresponding either to a feature value or to \(\texttt{unknown}\). Each feature $i \in \featureset$ has an associated \textit{cost}, $k_i$. A \textit{cost function} is defined as $\tau: \mathcal{F}_p \times \featureset \rightarrow \mathbb{Z}^{+}$, such that $\tau(f_p, i) = 0$ if $f_p = \{\texttt{unknown}\}$, and $\tau_f(f_p, i) = k_i \in \mathbb{Z}^{+}$, otherwise. Let $x = (x_1, \dots, x_n) \in \mathcal{F}_p$. The description cost function $\mathcal{T}: \mathcal{F}_p \rightarrow \mathbb{Z}^{+}$ is given by $\mathcal{T}(x) = \sum\limits_{i \in \featureset} \tau(x_i)$.
\end{definition}

We can now consider how agents interact: we denote every one-to-one interaction a conflict resolution dispute, or \textit{dispute}, for brevity. Agents will act with antagonistic intent: one of them pushes for a change in the status quo, and is thus called \textit{proponent}, while the other is called \textit{opponent}.

In such pairwise disputes between agents, we need a mechanism for defining what determines an \textit{objectively fair outcome} of a contest, given two descriptions of agents. Following the idea behind cultures, we assume there is a \textit{shared definition of fairness} across agents.\footnote{This makes no commitment on what is the ontological status of this shared definition, or whether it is deemed `correct.' We simply assume there is such a definition and that it is shared among the agents.} An objectively fair outcome disregards privacy, since it is the outcome which should be achieved under complete information.

As opposed to this perfect outcome, we also consider what an agent involved in a dispute \textit{considers} a fair outcome. This we define as the \textit{subjectively fair outcome}, according to one of the agents in the dispute, which is dependent on the information available to it. Any agent always has full information of oneself, but may only have partial information about the other party.

\begin{definition}[Fair Outcomes]
  We define the set of \textit{roles} in any dispute to be $R = \{pr, op\}$, two constant symbols denoting respectively ``proponent'' and ``opponent'', for the rest of the section. An \textit{objectively fair outcome} function $\omega: \mathcal{F} \times \mathcal{F} \rightarrow R$ receives the descriptions of the proponent and opponent agents, respectively, and decides which should win a dispute by returning the role of the winning agent in the dispute. %
  The \textit{subjectively fair outcome} function $\omega_p: R \times \mathcal{F} \times \mathcal{F}_p \rightarrow R$ receives as input whether the agent under consideration (the \textit{subject}) is the proponent or opponent, the description of the agent, and the private description of the adversary, and outputs the role of the agent which deserves to win the dispute.
\end{definition}

The choice of information disclosure to the adversary is determined by a \textit{disclosed information function} $\alpha$, which takes an agent (who is deciding which information to disclose), as well as the adversary, and produces a partial description $\mathcal{F}_p$ of the agent that fits under a specific privacy budget $g \in \mathbb{Z}^{+}$. 

\begin{definition}[Disclosed Information]
We say $\alpha: A \times A \times \mathbb{Z}^{+} \rightarrow \mathcal{F}_p$ is a \textit{disclosed information function} if and only if, for any agents $a, b \in A$ and privacy budget $g \in \mathbb{Z}^{+}$, $\alpha$ satisfies $\mathcal{T}(\alpha(a, b, g)) \leq g$.
\end{definition}

Once both agents made their decisions on how to generate their partial descriptions, the dispute resolution function returns the winner of a dispute with partial information on both sides.

\begin{definition}[Dispute Resolution]
A \textit{dispute resolution function} $\phi: \mathcal{F}_p \times \mathcal{F}_p \rightarrow R$ decides whether the proponent or the opponent wins the contest given their partial feature descriptions, respectively, by returning the role of winning agent in the dispute.
\end{definition}

The difference between functions $\omega$, $\omega_p$, and $\phi$ goes as follows: the first, $\omega$, returns the objectively fair outcome if all information is public and mutually known between both agents (complete). The function $\omega_p$ represents the perspective of \textit{the subject} upon receiving disclosed information from the adversary, after all, an agent always knows its full description but has partial information about other agents (depending on what is disclosed). Lastly, $\phi$ represents the final outcome given both agents' partial descriptions, after the dispute is resolved in some way.

While those definitions cover the essential concepts for pairwise (i.e. local) interactions, we still lack a relation of such pairwise interactions to a global state of the system. We assume the existence of a set $S$ of possible (global) states in which the system can be, with regards to the population of agents and a global notion of fairness. From an initial state, a set of local interactions transition the system into a final state. This is said to be the \textit{transition} of the system. 

\begin{definition}[System Transition Function]
  Let $A$ be the set of agents, $\alpha$ be a disclosed information function, $\phi$ the dispute resolution function, $g$ a privacy budget, and $S_0 \in S$ a state (referred as the \textit{initial state}). A \textit{system transition function} $\sigma$ is a function such that $\sigma(A, \alpha, \phi, g, S_0) \in S$, and we call $\sigma(A, \alpha, \phi, g, S_0) = S_F$ the \textit{final state}.
\end{definition}

We have introduced two aspects of fairness: perspective and scope. Perspective differentiates objective from subjective fairness, while scope differentiates global from local fairness. We formalise divergences in conflict resolution outcomes through the notion of a \textit{fairness loss}.  

\subsection{Fairness Loss}

The outcomes of privacy-restricted fairness disputes may disagree with our previously-defined notions of fair outcomes. We introduce \textit{fairness loss functions} as a means for comparing these.

The first definition for loss evaluates whether the dispute resolution outcome matches the objectively fair one, for a given privacy budget.

\begin{definition}[Objective Local Fairness Loss]
\label{definition:objective-local-loss}
The \textit{objective local fairness loss} function $l_{OL}: A \times A \times \mathbb{Z}^{+} \rightarrow \{0, 1\}$ is defined as 
\[
    l_{OL}(a, b, g) = 
        \begin{cases} 
          0, & \text{if } \phi(\alpha(a, b, g), \alpha(b, a, g)) = \omega(d(a), d(b)) \\
          1, & \text{otherwise}
        \end{cases}
\]
where $a, b \in A$ are agents, with $a$ the proponent and $b$ the opponent, and $g \in \mathbb{Z}^{+}$ denotes a privacy budget.
\end{definition}

Our second definition formalises whether the dispute resolution outcome is the same as the subjectively fair outcome for the agent playing the role $r$ in the dispute.

\begin{definition}[Subjective Local Fairness Loss]
\label{definition:subjective-local-loss}
The \textit{subjective local fairness loss} function $l_{SL}: A \times A \times R \times \mathbb{Z}^{+} \rightarrow \{0, 1\}$ is defined as 
\[
    l_{SL}(a, b, r, g) = 
        \begin{cases} 
          0, & \text{if } \phi(\alpha(a, b, g), \alpha(b, a, g)) = \omega_p(r, d(a), \alpha(b, a, g)) \\
          1, & \text{otherwise}
        \end{cases}    
\]
where $a, b \in A$, $r \in R$, and $g \in \mathbb{Z}^{+}$ denotes a privacy budget. %
\end{definition}

Finally, for objective and subjective global fairness, we characterise its requirements, but leave it to be specified application-wise. We present an applied experiment of objective and subjective global fairness in Section~\ref{section:boats}.

\begin{definition}[Global Fairness Loss]
  \label{definition:objective-global-loss}
  The \textit{objective global fairness loss} function $\Omega: S \rightarrow \mathbb{R}^{+}$, maps a state of a system to an unfairness value according to an objective notion of fairness. Analogously, a \textit{subjective global fairness loss} function $\Omega_p: S \rightarrow \mathbb{R}^{+}$ maps the same state to an unfairness value that stems from the subjective perception of the population. A higher value in either means that the state is less desirable, that is, less fair.
\end{definition}

We now introduce definitions pertaining to `orderly' cases of the previous definitions. Intuitively, a dispute resolution function is one such that when all information is available, an objective local fair outcome is achieved. When that is the case, we call it \textit{publicly sound}, formally as follows:

\begin{definition}
A dispute resolution function $\phi$ is \textit{publicly sound} iff for all $f_a, f_b \in \mathcal{F}$, $\phi(f_a, f_b) = \omega(f_a, f_b)$.
\end{definition}

We also define an agent being reasonable when, given perfect information about the other agent, $a$ always considers the objectively fair outcome a subjectively fair outcome, in any contest with another agent $b$. Being reasonable thus implies that the only deterrent towards an objectively fair outcome is partial information, as a reasonable agent will always agree with the objectively fair outcome when given enough information.

\begin{definition}
An agent $a \in A$ is \textit{reasonable} iff for any other agent $b \in A$, and any role $r \in R$, $\omega_p(r, d(a), d(b)) = \omega(d(a), d(b))$.
\end{definition}

Global outcomes can now connect with local outcomes by stipulating that, whenever complete information is available and the dispute resolution function is publicly sound, then the objective global fairness loss is 0. This corresponds to system transitions in which, whenever all disputes are resolved with objectively fair outcomes, then the final state is always objectively globally fair. 
\begin{definition}
  \label{definition:publicly-unbiased}
  Let $\alpha$ be a disclosed information function, $\phi$ be a publicly sound dispute resolution function, and $g \in \mathbb{Z}^{+}$ be such that, for every pair of agents $a,b \in A$, $\phi(\alpha(a, b, g), \alpha(b, a, g)) = \omega(d(a), d(b))$.
  A system transition function $\sigma$ is \textit{publicly unbiased} iff for any state $S_0 \in S$, we have that $\Omega(\sigma(A, \alpha, \phi, g, S_0)) = \Omega_p(\sigma(A, \alpha, \phi, g, S_0)) = 0$.
\end{definition}

Finally, we can state a result, thus proving that:

\begin{theorem}
\label{theorem:zero-loss}
If $\phi$ is a publicly sound dispute resolution function, every agent in $A$ is reasonable, $g \in \mathbb{Z}^{+}$ is such that for all $a \in A$, $\alpha(a,g) = d(a)$, and $\sigma$ is publicly unbiased, then, for all $a, b \in A$ and state $S_0 \in S$, $l_{OL}(a, b, g) = l_{SL}(a, b, g) = \Omega(\sigma(A, \alpha, \phi, g, S_0)) = \Omega_p(\sigma(A, \alpha, \phi, g, S_0)) = 0$. 
\end{theorem}
\begin{proof}
Considering Def.~\ref{definition:objective-local-loss}, we know that for a privacy budget $g$ high enough such that $\alpha(a,g) = d(a)$ and $\alpha(b,g) = d(b)$, we achieve a global loss of $0$. We can then replace the private descriptions $\alpha(a,g)$ and $\alpha(b,g)$ with their public equivalents $d(a)$ and $d(b)$ respectively, obtaining $\phi(d(a), d(b)) = \omega(d(a), d(b))$. Since $\phi$ is publicly sound, this holds.

A similar reasoning applies to $l_{SL}$ in Def.~\ref{definition:subjective-local-loss}. Suppose a $g$ high enough such that $\alpha(a,g) = d(a)$ and $\alpha(b,g) = d(b)$, for every $a, b \in A$. Applying this to $l_{SL}$, for any $o \in \{a, b\}$, the value is 0 if $\phi(d(a), d(b)) = \omega_p(o, d(a), d(b))$. This holds, since $a$ and $b$ are reasonable agents.

Finally, for $\Omega$, it is a direct consequence of Def.~\ref{definition:publicly-unbiased}. Due to the constraint on $g$, $\alpha(a, b, g) = d(a)$, and similarly for $b$. Thus $\phi(\alpha(a, b, g), \alpha(b, a, g)) = \omega(d(a), d(b))$ is satisfied since $\phi$ is publicly sound.
\end{proof}

With the extant formalisms providing a frame of reference and context, we introduce a problem statement that predicates the motivation and guides our resulting contributions in this paper.

\textbf{Problem Statement.} Let $A$ be a set of agents. We assume all agents in $A$ are self-interested and will experience mutual conflicts of interest regarding contended resources. We assume an initial state $S_0$, which transitions depending on outcomes of agent conflicts. For every conflict between any agents $a, b \in A$ that arises over a contended resource, agents will engage in a dispute with mutually-exchanged partial information. This will be governed by means of: an equal privacy budget $g$ for both agents, a disclosed information function $\alpha$ representing the \textit{strategy} of both agents for releasing information, and a dispute resolution function $\phi$. For any given $g$, find an $\alpha$ that minimises $\Omega$ and $\Omega_p$.%

\section{Empirical Analysis}
\label{section:random-experiments}

In this section, we propose an experiment to measure the effectiveness of distinct dialogue strategies with respect to global and local fairness under varying degrees of privacy. As shown by Theorem~\ref{theorem:zero-loss}, both objective and subjective fairness loss definitionally converge to zero if the conflict resolution mechanism is publicly sound and agents act reasonably. However, this is a loose bound -- as those definitions do not account for the specific mechanics of dialogues, conflict resolution, or fairness.

Consequently, we use the architecture proposed in Section~\ref{section:architecture} to instantiate scenarios with randomly-generated agents and cultures. Our interest lies in empirically observing the impact on global and local fairness under different privacy budgets using four different argumentation strategies. Below, we introduce the details for our experimental setup and evaluation mechanisms.

\subsection{Setup}
\label{section:random-setup}
Let $A = \{q_1, q_2, \dots, q_{|A|}\}$ be a finite set of ${|A|}$ agents. Every agent $q \in A$ possesses an $m$-tuple $d(q) = (i_1, i_2, \dots, i_m)$ for all $i \in I$, where $|I| = m$, representing that agent's \textit{feature description}, i.e., its internal traits and characteristics. The value function $\mu: A \times I \rightarrow \mathbb{Z}^+$ returns the numerical value of a feature $i \in I$ for an agent $q \in A$.

Let $C = \{\mathcal{A}, \mathcal{R}, \mathcal{K}\}$ be a culture, where $\mathcal{A} = \{\gamma, a_1, a_2, \dots, a_m\}$ is composed of $m + 1$ arguments (%
with a single motion $\gamma \in \mathcal{K}$). Let $\textit{index}(a_j) = j, a_j \in \mathcal{A}$ denote the index of an argument. We generate $\mathcal{R}$ randomly, satisfying the conditions: \textit{i)} the underlying $AF = (\mathcal{A}, \mathcal{R)}$ has exactly one connected component; \textit{ii)} for every $(a, b) \in \mathcal{R}$, $\textit{index}(a) > \textit{index}(b)$. 

Non-motion arguments in $C$ will represent a feature comparison between two agents. We consider the alteroceptive expansion of $C$, $C_x = (\mathcal{A}_x, \mathcal{R}_x)$ (Def.~\ref{definition:alteroceptive}). Every verified-fact argument $a_{F}^{q_j} \in \mathcal{A}_x$ is associated to a verifier function $v_a(q_j, q_k) = \texttt{True}$ if $\mu(q_j, \textit{index}(a)) > \mu(q_k, \textit{index}(a))$; otherwise $\texttt{False}$, for $q_j, q_k \in A$. All hypotheses are trivially associated to verifier functions that always return \texttt{True}.

Informally, this means that every feature $i \in I$ is represented in a dialogue between two agents by the respective hypotheses and verified-facts regarding which agent has a superior value in feature $i$. This abstractly represents any potential feature in a system.

\subsection{Evaluation Metrics} \label{subsection:evaluation}

To empirically study the impact of privacy on local and global fairness using different argumentation strategies, we define two evaluation metrics. First, we consider the \textit{aggregated subjective local unfairness} as the summation of the subjective local fairness loss $l_{SL}$ over all pairs of distinct agents. In our case, this is modelled by dialogues that end due to a lack of privacy budgets from one of the agents, as noted by dispute result \textit{(ii)} in Section~\ref{section:architecture}.

To measure \textit{objective local fairness loss}, we calculate a \textit{ground truth} of all pairwise interactions of agents. This ground truth is defined as an $|A| \times |A|$ matrix $\mathbf{GT}$, with entries being elements of $R = \{pr, op\}$. The entries of $\mathbf{GT}$ are the \textit{objectively fair outcome} of the dispute in which agents $a_j$ and $a_k$ assume the roles of $pr$ and $op$, respectively, that is: $\mathbf{GT}_{j,k} = \omega(d(q_j), d(q_k))$.

Our instantiation of the \textit{objectively fair outcome function} $\omega$ is defined the following procedure:
\begin{enumerate}
    \item Given the complete descriptions $d(pr)$ and $d(op)$, run all verifier functions for all arguments and remove all \textit{verified-fact} arguments that return \texttt{False}.
    \item Check for \textit{sceptical acceptance}\footnote{We used the $\mu$\texttt{-toksia} \citep{Niskanen2020-toksia:Reasoner} SAT-based solver for the simulations.} of the proponent's motion $\gamma_H^{pr}$. If yes, then return $pr$. Return $op$ otherwise.
\end{enumerate}

The resulting \textit{ground truth} can also be represented as a digraph $G_{\mathbf{GT}} = (A, E)$, where for every two distinct agents $q_j, q_k \in A$, we say an arc $(q_j, q_k) \in E$ iff $\mathbf{GT}_{j,k} = pr$ and $\mathbf{GT}_{k,j} = op$. We generalise this definition for any strategy $\alpha$ and denote the digraph $G_{\alpha}$ as a \textit{precedence graph} of agents.

Objective global fairness loss compares the ground truth precedence graph, $G_{\mathbf{GT}}$, to the precedence graph resulting from a strategy, $G_{\mathbf{RE}(g, \alpha)}$. To compare these two precedence graphs, we use the DAG dissimilarity metric seen in \citep{Malmi2015BeyondGraphs}, defined below.

Let $G_1 = (A, E_1)$ and $G_2 = (A, E_2)$ be two precedence graphs. Let $e$ denote an arc $(q_j, q_k)$. Let $c_1$ denote the number of occurrences where $e \in E_1$ and its reverse $(q_k, q_j) \in E_2$, $c_2$ as the number of occurrences where $e$ exists in either $E_1$ or $E_2$ but not the other, and $c_3$ as the number of occurrences where neither $E_1$ or $E_2$ contain $e$. The DAG dissimilarity between two graphs $G_1 = (A, E_1)$ and $G_2 = (A, E_2)$ is $K(G_1, G_2, y_1, y_2) = c_1 + c_2y_1 + c_3y_2$, where $0 \leq y_2 < y_1 \leq 1$ are constants. We choose $y_1 = 2/3$ and $y_2 = 1/3$, via integration.\footnote{We decide to get a representative value from the set of possibilities $0 \leq y_2 < y_1 \leq 1$, so we use the \textit{centroid}, found by the analytical solution to the integration of $K$ over $y_1$ from $0$ to $1$, and $y_2$ from $0$ to $y_1$.}

\subsection{Strategies}
\label{subsection:strategies}
Let $\mathcal{D}$ be the set of all possible privacy-aware dialogues. Let $D = m_0, \dots, m_n$ be a privacy-aware dialogue, with $m_n = (w,a)$ the last movement used. We define $\eta : \mathcal{D} \rightarrow 2^{\mathcal{A}_x}$ as the function that returns the set of all arguments $r \in \mathcal{A}_x$ such that $r$ attacks $a$ and $m_{n+1} = (\overline{w},r)$ can be used as the next move in the dialogue, that is, $m_0, \dots, m_n, m_{n+1}$ is also a privacy-aware dialogue. We enumerate 4 strategies for choosing a rebuttal argument $r \in \eta(D)$.

\begin{enumerate}
    \item \textit{Random:} $r$ is sampled randomly with equal probability for all $r \in \eta(D)$.
    \item \textit{Minimum Cost:} $r$ is the argument in $\eta(D)$ with lowest cost.
    \item \textit{Offensive:} $r$ is the argument in $\eta(D)$ that attacks most other arguments in $\mathcal{A}_x$.
    \item \textit{Defensive:} $r$ is the argument in $\eta(D)$ that suffers the least attacks in $\mathcal{A}_x$.
\end{enumerate}

Let $\alpha \in \mathcal{S}$ be a strategy, where $$\mathcal{S} = \{\texttt{random}, \texttt{min\_cost}, \texttt{offensive}, \texttt{defensive}\}.$$ A \textit{result} is defined as an $n\times n$ matrix $\mathbf{RE} : \mathbb{Z}^+ \times \mathcal{S} \rightarrow R^{n \times n}$, where $R = \{pr, op\}$. Every item $\mathbf{RE}_{j,k} = \phi(\alpha(q_j, q_k, g), \alpha(q_k, q_j, g))$, for $\mathbf{RE}_{j,k}$ represents the \textit{dispute resolution outcome} of a dispute between agents $q_j$ and $q_k$, assuming the roles of $pr$ and $op$, respectively.
We instantiate $\phi$ by carrying out a dialogue between $q_j$ and $q_k$ until a winner is found, returning $pr$ or $op$ accordingly.
Analogously to the ground truth model, we can generate a precedence graph $G_{\mathbf{RE}(g, \alpha)} = (A, E)$, where for every two distinct agents $q_j, q_k \in A$, we say an arc $(q_j, q_k) \in E$ iff $\mathbf{RE}_{j,k} = pr$ and $\mathbf{RE}_{k,j} = op$.

\subsection{Simulations}
\label{subsection:random-experiments}
Each experiment consists of a set of $|A| = 50$ agents with fixed feature values initialised randomly. For each experiment, we create a random acyclic culture $C = (\mathcal{A}, \mathcal{R}, \mathcal{K})$ with $|\mathcal{A}| = 50$, $|\mathcal{R}| \simeq 400$, and $\mathcal{K} = \{\gamma\}$. We generate an \textit{alteroceptive expansion} $C_x$ of $C$ where the privacy cost $\tau(a_x)$ for all $a_x \in \mathcal{A}_x$ is randomly initialised as $1 \leq \tau(a_x) \leq 20$. Using privacy-aware dialogue games as dispute resolution mechanisms, we generate a precedence graph $G_{\mathbf{RE}(g, \alpha)}$ for each strategy $\alpha \in \mathcal{S}$.

We repeat each experiment for integer value privacy budgets $0 \leq g < 60$. We aggregate measures of global and local fairness over 1900 trials,\footnote{Total simulation time exceeded 55 hours on an \texttt{i7-8550U} CPU with 16GB of RAM.} where each trial randomly initialises all agents' feature descriptions and the culture. 

The plots in Figure~\ref{fig:local_global} show the average subjective local fairness loss (left) and the average objective global fairness loss (right) for all four strategies over a range of privacy budget values. All strategies converge to zero subjective local fairness loss given high enough privacy budgets. However, at a given privacy budget value, the defensive strategy dominates the other strategies. Similarly, the defensive strategy is dominant with respect to objective global fairness. In addition to achieving a lower objective global fairness loss at a given privacy budget value, the convergence value of the defensive strategy is lower than that of the other strategies.

\begin{figure}
\includegraphics[width=0.49\linewidth]{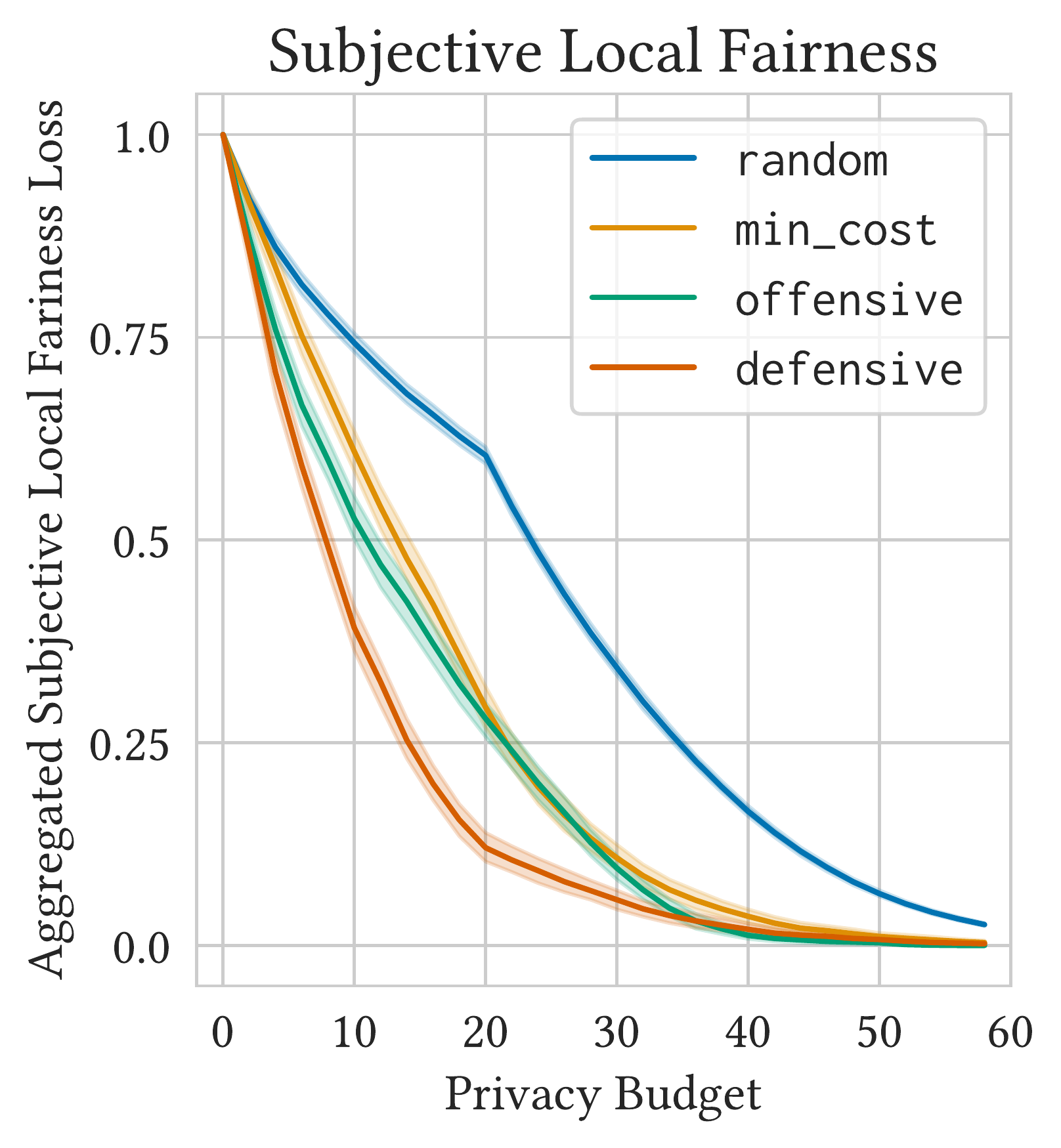}
\includegraphics[width=0.49\linewidth]{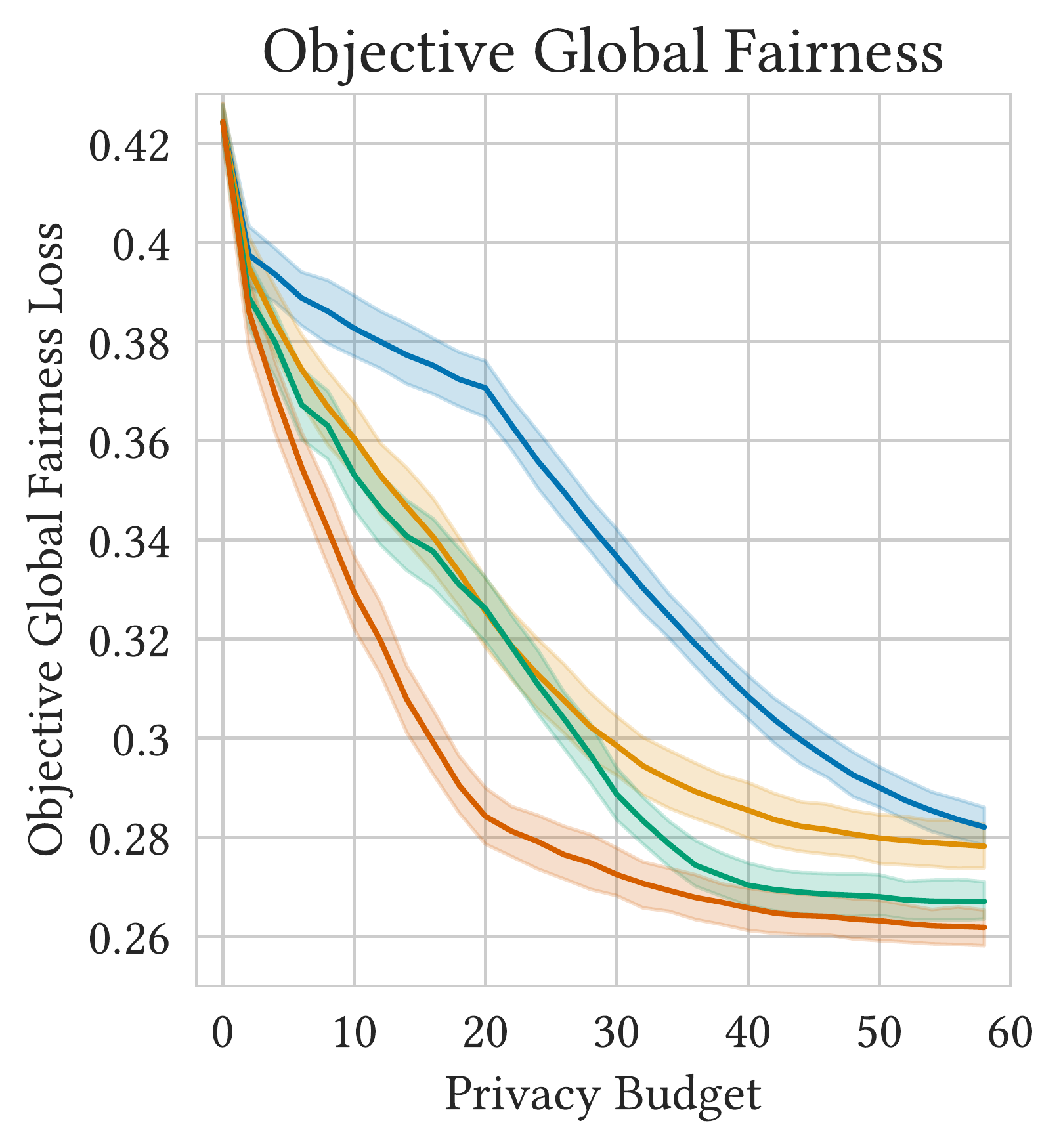}

\caption{The defensive strategy outperforms the other strategies with respect to both local and global fairness under privacy constraints. We show the average subjective local fairness loss (left) and the average objective global fairness loss (right) across a range of privacy budgets for \texttt{random}, \texttt{min\_cost}, \texttt{offensive}, and \texttt{defensive}. The shaded regions show 99\% confidence intervals on the mean values. Note that the small change in behaviour for $g \geq 20$ is due to the fact that $\max \tau(a_x) = 20$ (see Section~\ref{subsection:random-experiments}).}
\label{fig:local_global}
\end{figure}

\subsection{Privacy Efficiency}
\label{subsection:privacy-efficiency}

We run approximately 10 million dialogues between 95,000 different agents, using the same 1900 random cultures from the previous experiment. However, this time we observe how much privacy cost was used by each strategy to finish the dialogues. Agents were free to extend their dialogues for as long as required. Figure~\ref{fig:privacy-efficiency-random} shows an empirical cumulative distribution function of the proportion of dialogues with cost $z$, where $z \geq z'$, s.t.~$l_{SL}(q_i, q_j, r, z') = 1$, for any $\{q_i, q_j\} \subset A$, $\alpha \in \mathcal{S}$, and any $r \in \{pr, op\}$ (i.e. dialogues that need a privacy budget higher than $z$ to not be cut short by it). This illustrates the effect of different strategies in minimising privacy cost, even in unrestricted dialogues. Results are consistent with the findings in Section~\ref{subsection:random-experiments}.

\begin{figure}
\centering
\includegraphics[width=0.7\linewidth]{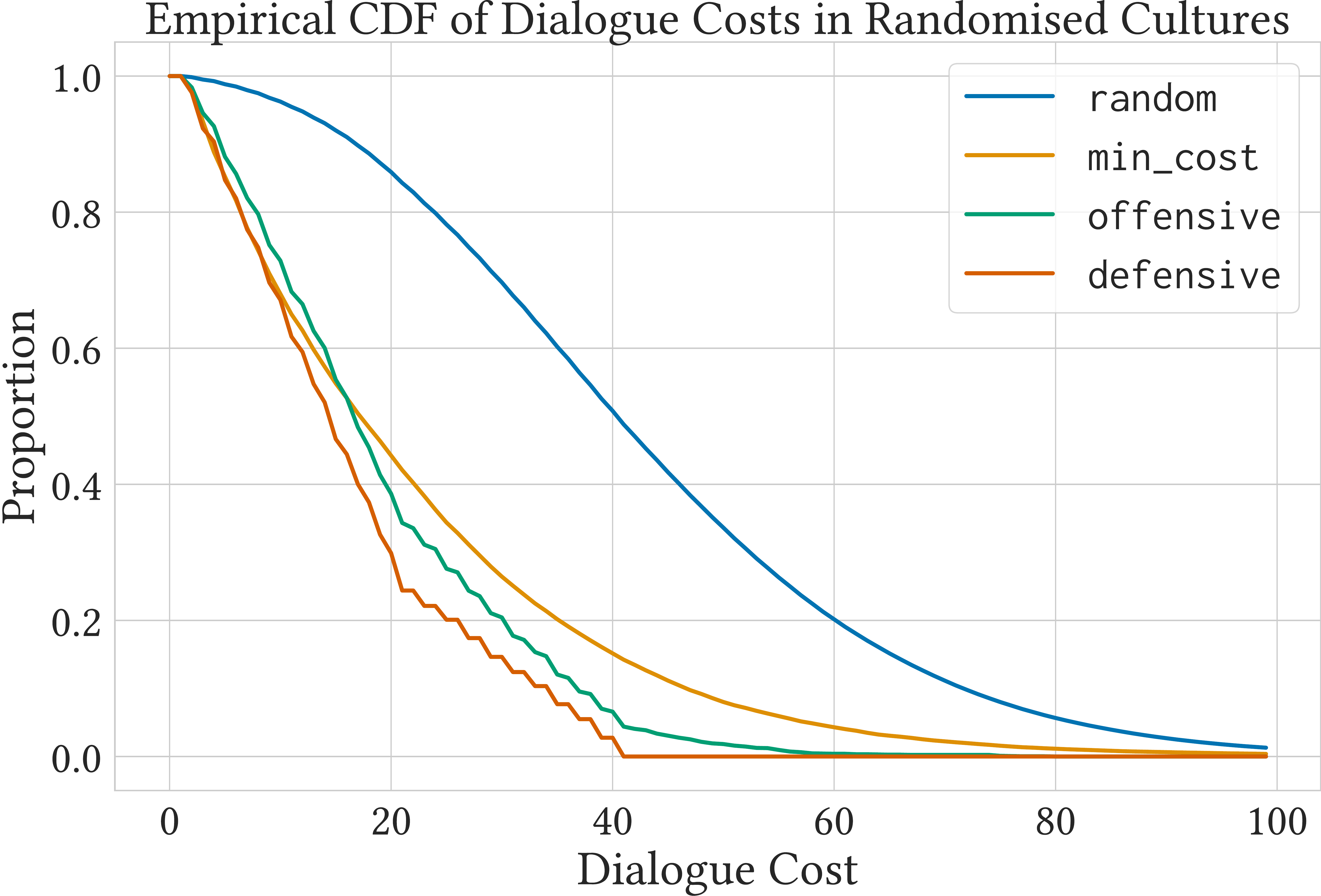}
\caption{An empirical cumulative distribution function of privacy costs for unrestricted dialogues in over 1900 different random cultures.}
\label{fig:privacy-efficiency-random}
\end{figure}

\section{Multi-agent Application}
\label{section:boats}

Manifestly, the previous empirical analysis explores important, but essentially \textit{abstract} aspects of the problem. Observing properties under a spectrum of multiple different privacy budgets and with large randomly-generated cultures allows us to observe a space of multiple possible systems. However, realistic applications are not likely to explore an assortment of privacy budgets, as the limits on privacy are likely to be fixed or seldom vary. In like manner, cultures can be explainable representations of real-world rules and their exceptions - and much like in the privacy budgets' case, these are also unlikely to fluctuate in applications predicated in reality.

On these grounds, we build on this conceptual foundation to demonstrate how said aspects may also materialise in an applied setting, even under assumptions of: \textit{i)} a single fixed privacy budget, and \textit{ii)} a single fixed alteroceptive culture. Drawing inspiration from the \textit{`Busy Barracks'} game seen in \citep{Raymond2020Culture-BasedDeconfliction}, we apply our architecture to a multi-agent simulation of speedboats (see Figure~\ref{fig:simulator}). 

\begin{figure}
\includegraphics[width=\linewidth]{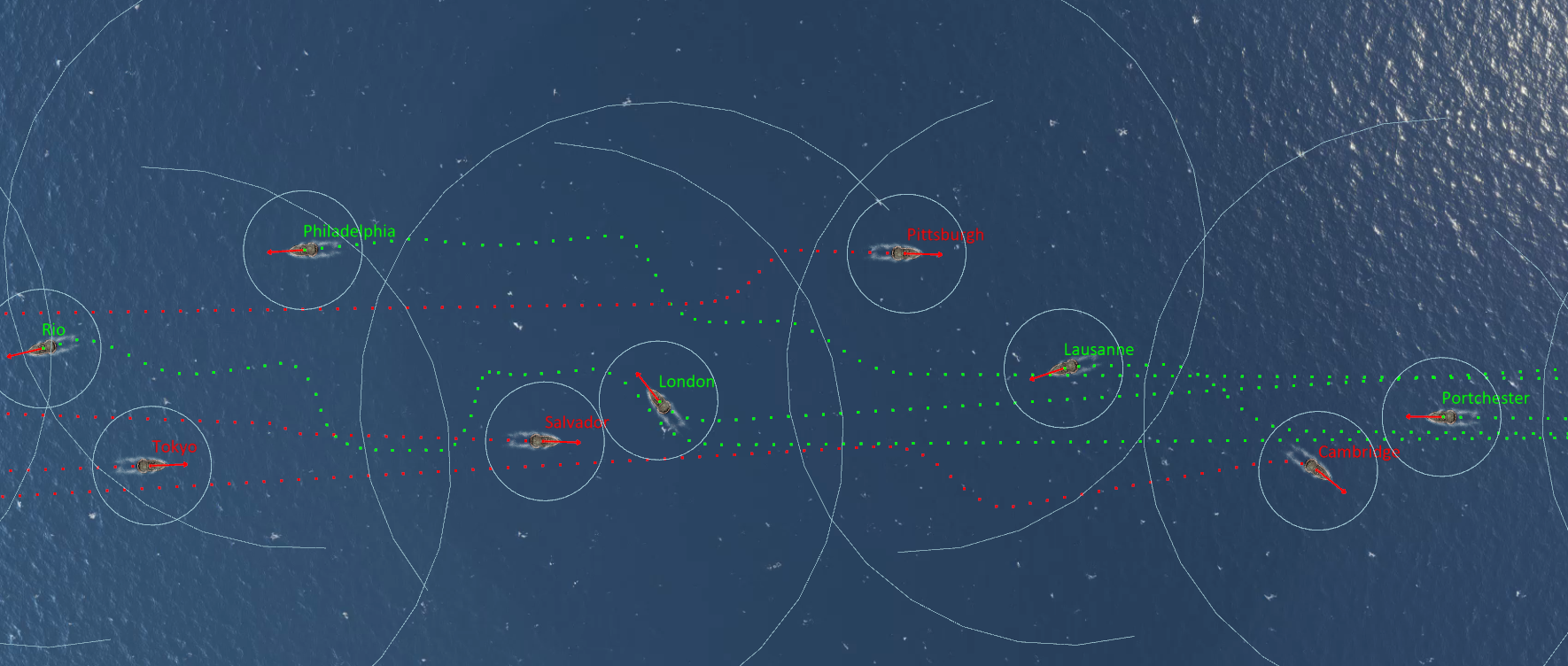}
\caption{Screenshot of our multi-agent collision avoidance simulation. Speedboats are required to engage in privacy-aware dialogues to resolve who has right-of-way before getting into the risk of collision. Boat and ripple art from \citep{Felvegi2013ShipsEffect}.}
\label{fig:simulator}
\end{figure}

\subsection{Simulator}

Our environment consists of a 20km long $\times$ 2km wide two-dimensional water surface, without any obstacles. In every trial, 16 agents are instantiated in a `head-on parade' scenario, as follows:
\begin{itemize}
    \item 8 out of the 16 agents start from the west, and 8 others start from the east.
    \item Every agent starts with at least 1km of longitudinal separation from any other agent.
    \item Their start and goal vertical coordinates are randomised within $\pm$200m from the vertical midpoint.
    \item Every agent's destination is at the opposite side of the map, and they are initialised with an exact heading towards it. If undisturbed, any agent should be able to execute a perfect straight line towards their destination without any lateral movement or rotation.
\end{itemize}

Speedboats are simulated with a physics model adapted from \citep{Monster2003CarGames} and \citep{Linkovich2016Carphysics2d}, where elements such as water resistance, drag, and mass are taken into consideration when calculating frames. We simulate and control each boat at 20Hz refresh rate. Each trial spans around 10-20 minutes of simulation at 20Hz for 16 agents, and we collect around 400,000 telemetry data points per trial, including velocity, lateral acceleration, yaw rate, and lateral jerk. We run 900 trials and generate over 40GB of trajectory and telemetry data for our analyses.

The properties of the vehicle are coarsely modelled after the real-life high performance passenger boat Hawk 38 \citep{Sunseeker2019SunseekerBrochure}, with capacity for 7 seated passengers, dual 400hp motors, and top speed of 30m/s (approx. 60 knots). It is assumed that vehicles will drive at maximum speed whenever possible.

\subsection{Boat Culture}

Akin to Section~\ref{section:random-setup}, each agent is bestowed with properties extracted from a culture, initialised according to a given distribution. Table~\ref{table:boat-culture} illustrates the culture designed for the experiment. To preserve a certain degree of realism, the initialisation of each agent is still random, but respecting certain reasonable restrictions (e.g. civilians cannot have high military ranks or engage in combat, spies are always civilian or corporate, etc).
\begin{table}[]

\begin{tabularx}{\textwidth}{|l|l|l|X|l|}
\hline
$i$  & \textbf{Property}            & $\tau(i)$      & \textbf{Possible values} of $\mu$ \textit{(ascending order of importance)}                                         & \textbf{Attacks}               \\ \hline
0  & motion                & $0$            & -                                                                                                & -                     \\ \hline
1  & VehicleAge            & $4$            & \texttt{\{new, used, worn, old, vintage\}}                                                       & \{0\}                 \\ \hline
2  & VehicleCost           & $10$           & \texttt{\{cheap, ok, expensive, very\_expensive, millions\}}                                     & \{0, 1\}              \\ \hline
3  & HigherCategory        & $0$            & \texttt{\{civilian, corporate, police, coast\_guard, military\}}                                 & \{0, 1, 2\}           \\ \hline
4  & TaskedStatus          & $3$            & \texttt{\{at\_ease, returning, tasked\}}                                                         & \{0, 1, 2, 3\}        \\ \hline
5  & PayloadType           & $5$            & \texttt{\{empty, food, medical\_supplies\}}                                                      & \{0, 1, 2\}           \\ \hline
6  & TaskNature            & $7$            & \texttt{\{leisure, sport, trade, training, patrol, pursuit, combat\}}                            & \{4, 5\}              \\ \hline
7  & VIPOnBoard            & $13$           & \texttt{\{ordinary\_person, business\_person, celebrity, politician\}}                           & \{0, 1, 2, 4\}        \\ \hline
8  & MilitaryRank          & $8$            & \texttt{\{no\_rank, officer, lieutenant, commander, captain, major, colonel, general, admiral\}} & \{3, 5, 6, 7\}        \\ \hline
9  & DiplomaticCredentials & $12$           & \texttt{\{no\_credentials, diplomat, united\_nations\}}                                          & \{0..8\}              \\ \hline
10 & SensitivePayload      & $15$           & \texttt{\{no\_sensitive\_payload, weapons, wanted\_prisoner\}}                                   & \{0..9\}              \\ \hline
11 & UndercoverOps         & $20$           & \texttt{\{no\_spy, spy\}}                                                                        & \{3, 4, 6, 7, 8, 10\} \\ \hline
12 & EmergencyNature       & $10$           & \texttt{\{no\_emergency, mechanical, sick\_passenger, fire\}}                                    & \{0..11\}             \\ \hline
13 & SuperVIPOnBoard       & $16$           & \texttt{\{no\_super\_vip, prime\_minister, head\_of\_state\}}                                    & \{0..12\}             \\ \hline
\end{tabularx}%
\caption{Properties present in our culture, along with their possible values for each agent. The `attacks' column indicates which other properties can be defeated by the respective row. }
\label{table:boat-culture}
\end{table}

\subsection{Dialogues and Collision Avoidance}

In our simulation, agents have no prior knowledge of other agents' descriptions. Properties with $\tau = 0$ can be communicated instantaneously. It is assumed that agents act lawfully and will not lie, but they are nonetheless self-interested and will fight for the right of way in every conflict that arises. In our scenario, we establish that more sensitive information requires further \textit{time to obtain clearance} before communicating to another agent. Therefore, in this experiment, \textit{privacy costs are associated with time} (and consequently, distance). We model their avoidance mechanism using a modified artificial potential field algorithm \citep{Warren1990MultipleFields}.

Let $q_i, q_j \in A$ be any two agents, and $dist: A \times A \times \mathbb{Z} \rightarrow \mathbb{Z}^+$ be a function that determines the euclidean distance between two agents at a specific time $t$. We define two constants $r_\text{max} = 1000$ and $r_\text{crit} = 100$ as the maximum effect and the critical minimum radii of a potential field. In our method, when $dist(q_i, q_j, t) > r_\text{max}$ and $dist(q_i, q_j, t+1) \leq r_\text{max}$ (i.e., the agents just got within the maximum effect radius for the first time), $q_i$ and $q_j$ will initiate a dialogue game $\phi(q_i, q_j, g)$ to decide who has right of way, where $g$ is the maximum privacy budget for each agent.

\begin{figure}[htb]
\centering
\includegraphics[width=\linewidth]{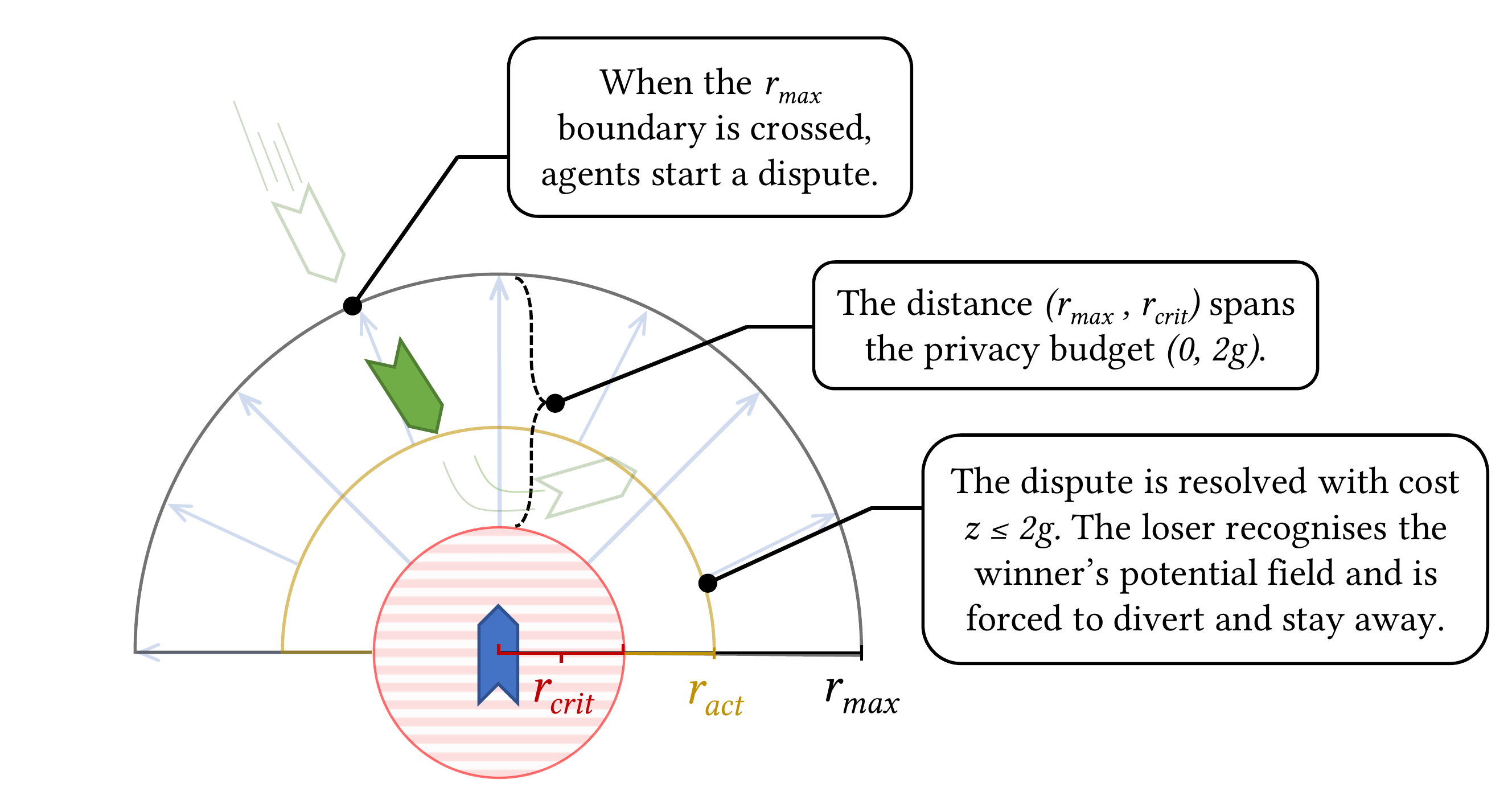}
\caption{Illustration of the collision avoidance mechanism from Definition~\ref{definition:activation-radius}.}
\label{fig:potential-field}
\end{figure}

\begin{definition} [Activation Radius]
Starting from the centre of each agent, we divide the space between $r_\text{max}$ and $r_\text{crit}$ into $2g$ concentric and uniformly-expanding rings. Let $\mathcal{T}(q)$ be the description cost of an agent, that is, how much of its privacy budget was spent to reveal information. Without loss of generality, suppose $q_i \in \{q_i, q_j\}$ is the winner of the dialogue game and $z = \mathcal{T}(q_i) + \mathcal{T}(q_j), z \leq 2g$ is the total combined privacy cost for that dialogue, considering both agents. The activation radius $$r_\text{act} = r_\text{max} - \frac{z}{2g} (r_\text{max} - r_\text{crit})$$ is the distance where the potential field (with full radius $r_\text{max}$) will be enabled for the first time for the losing agent, and will remain in effect until the end of the simulation.
\label{definition:activation-radius}
\end{definition}

In plain language, agents will start the dialogue as soon as they cross the $r_\text{max} = 1000$ boundary, and their cumulative privacy cost $z$ (e.g. the \textit{time} the dialogue takes) will determine how early or how late they ultimately decide which agent should have right of way (see Figure~\ref{fig:potential-field}). Higher values of $z$ will lead to late, aggressive evasive manoeuvres. In case $dist(q_i, q_j) < r_\text{crit}$, the \textit{winner} is also forced to divert to avoid a collision. In an ideal scenario, most agents will complete their dialogues in full with very low privacy costs and execute early, smooth manoeuvres to stay clear of other agents. Conversely, if dialogues are wasteful or activation radii are closer to $r_\text{crit}$ than $r_\text{max}$, then one of the following may happen in the \textit{local} scope, for any $q_i, q_j \in A$ and $r \in \{pr, op\}$: 

\begin{itemize}
    \item the evader makes a very late and aggressive evasion turn \textit{(low privacy efficiency: high $\mathcal{T}(\alpha(q_i, q_j, g))$)};
    \item the late evader accepts defeat but breaks into the critical radius of the winner and forces them out of their rightful trajectory \textit{(objective unfairness: $l_{OL}(q_i,q_j,g) = 1$ and $l_{SL}(q_i, q_j, r, g) = 0$)}; 
    \item the right agent wins the dialogue game, but the losing agent is not fully convinced all their reasons are covered \textit{(subjective unfairness: $l_{OL}(q_i,q_j,g) = 0$ and $l_{SL}(q_i, q_j, r, g) = 1$)};
    \item the wrong agent wins the dialogue game and forces the right one out of their way since they had no budget remaining to get to the correct decision \textit{(subjective and objective unfairness: $l_{OL}(q_i,q_j,g) = 1$ and $l_{SL}(q_i, q_j, r, g) = 1$)}.
\end{itemize}

As per our previous abstract experiment, we demonstrated that different strategies for choosing arguments (explanations) during the dialogue game leads to varied levels of performance with regards to privacy efficiency, subjective fairness, and objective fairness. We will perform simulations using the same strategies seen in Section~\ref{subsection:strategies} in the sections below.

\subsection{Privacy Efficiency Experiments}

Since this experiment involves a real simulation of a multi-agent system, its magnitude is significantly smaller in terms of number of dialogues and agents. We repeat the experiment in Section~\ref{subsection:privacy-efficiency}, this time with only one culture. We observe the privacy costs of 48,000 dialogues between the previous set of 1600 randomly-initialised agents for each strategy without any restrictions in privacy.

Figure~\ref{fig:privacy-efficiency} shows the effect of different strategies in minimising privacy cost for \textit{unrestricted} dialogues in the Boat Culture. Two phenomena can be observed: the greedy strategy for minimising cost (\texttt{min\_cost}) turns out to be the worse at that. This is due to the fact that, as privacy costs are not randomly determined this time, privacy costs are correlated with argument relevance/power, making stronger arguments more expensive. Additionally, both \texttt{offensive} and \texttt{defensive} exhibit similar performance. In a similar nature, in this culture, stronger arguments simultaneously attack many others, whilst being relatively unattacked. 

\begin{figure}
\centering
\includegraphics[width=0.7\linewidth]{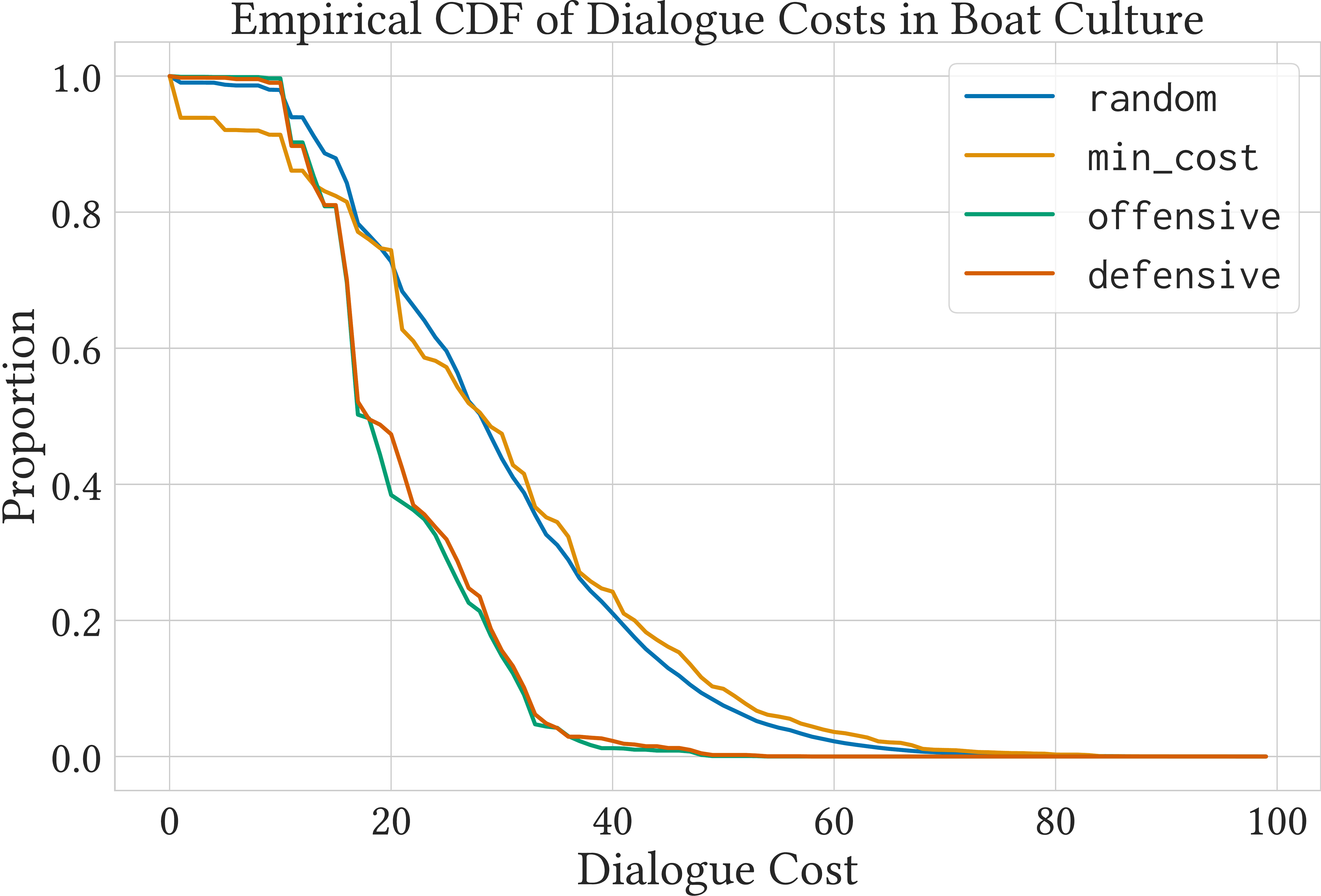}
\caption{Privacy costs for unrestricted dialogues in our instantiated Boat Culture (as shown in Table~\ref{table:boat-culture}). Note that some behaviours are different from Figure~\ref{fig:privacy-efficiency-random}. Contrary to the randomised case, where over 2000 different cultures are tested, this plot represents dialogues under a single culture.}
\label{fig:privacy-efficiency}
\end{figure}

\subsection{Ride Quality}

A more applied perspective on the same matter is to observe the status of vehicles across their trajectories. In our simulation, most disturbances and evasions take place when both sides of the parade meet in the middle. These conflicts generate peaks in lateral acceleration and other measurements. When those conflicts are resolved, agents usually follow trouble-free paths to their destinations.
We are interested in quantifying a general notion of `ride comfort' (supposing an acceptable limit of `comfort' for a powerboat at maximum speed) and minimising the `ride roughness' for the fictitious average passenger on board. We discard the initial acceleration and final braking data points, as we are only interested in changes in acceleration caused by other agents.

To do so, we integrate the area under each of those metrics per agent, and compute the mean of each agent's integrals per trial. We run 100 trials per strategy with a combined privacy budget value of $2g = 60$ and aggregate these values. We associate higher values of lateral acceleration and jerk represent as strong components in a passenger's experience of discomfort \citep{Nguyen2019InsightSingapore}. 

The first row in Figure~\ref{fig:global-results} shows the distribution of lateral acceleration, yaw rate, and lateral jerk measurements over 1600 simulations (16 agents x 100 trials) per strategy. We observe significantly lower values from \texttt{offensive} and \texttt{defensive} in all metrics. The \texttt{random} strategy also exhibited significantly lower values of lateral jerk compared to \texttt{min\_cost}.

\subsection{Subjective and Objective Trajectories}

\begin{figure}
\includegraphics[width=\linewidth]{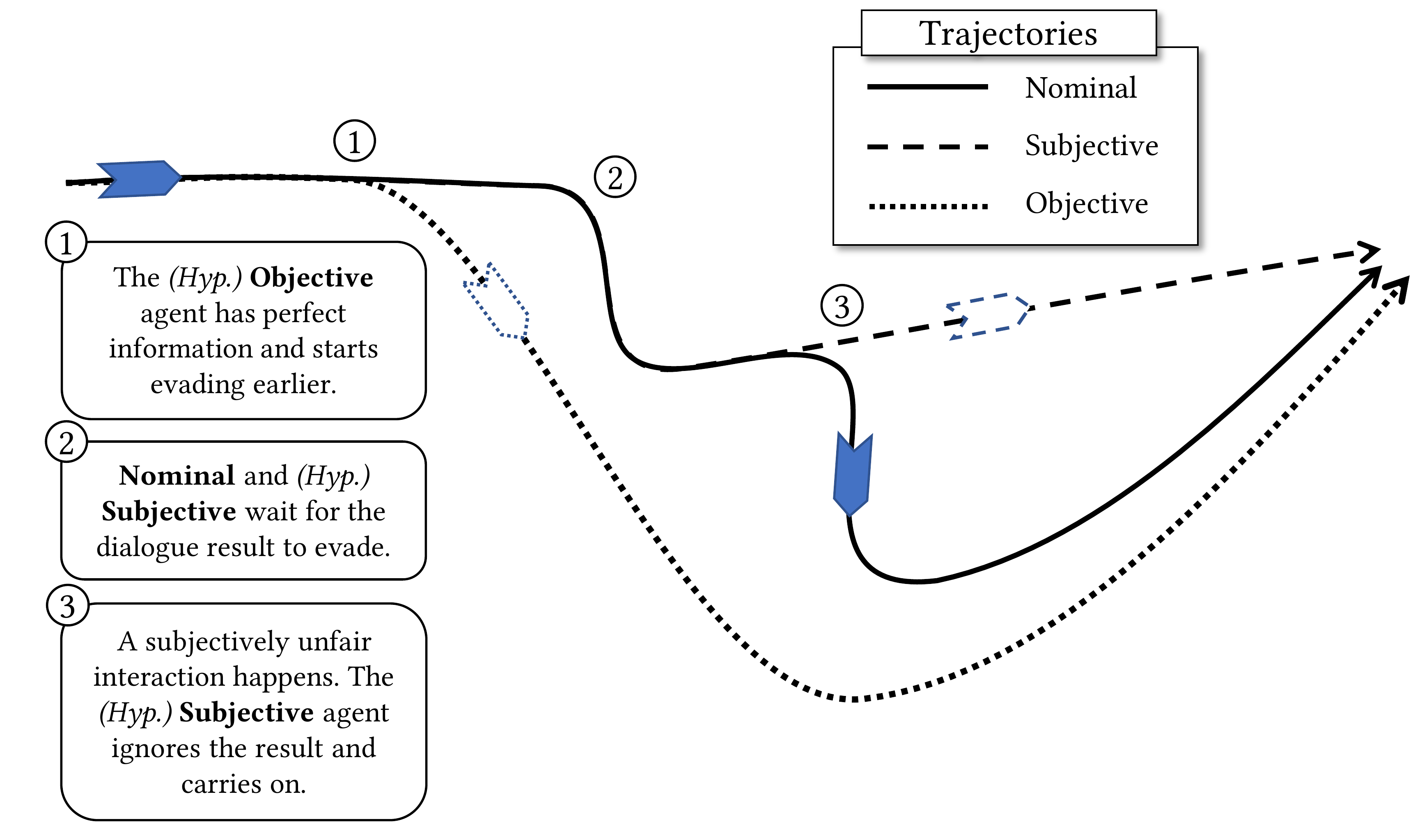}
\caption{An example of a Nominal trajectory and its hypothetical Subjective and Objective scenarios. The full line (Nominal) represents the actual trajectory performed by the agent in their trial. The Objective hypothesis (dotted line) can perform smoother manoeuvres and evade much sooner due to perfect information being immediately available at $r_\text{max}$. The Subjective hypothesis (dashed line) behaves like the Nominal trajectory, but will not concede upon subjective unfairness cases.}
\label{fig:compare-trajec}
\end{figure}

After observing the impact of privacy on the \textit{quality} of the ride, we now measure agents' perspectives in a \textit{global} scope. Instead of counting local instances of objective and subjective unfairness, we generate entire trajectories that cater to specific requirements. From this point, we will apply the following nomenclature:

\begin{itemize}
    \item \textbf{Nominal} trajectory ($\mathcal{J}_N$): this is the actual trajectory executed by the agent in a regular trial. Agents perform their dialogues normally and give way in case of defeat (as shown in Figure~\ref{fig:potential-field}), according to the current strategy mutually adopted by all agents.
    \item \textit{(Hypothetical)} \textbf{Subjective} trajectory ($\mathcal{J}^h_{S}$): same as Nominal, but whenever an agent loses a dialogue that it deems unfair, it hypothesises what would happen if the agent ignored the potential field of the opponent and forcefully drove straight through. In this case, agents will never `agree to disagree' and will \textit{not} give way if the dialogue ends due to privacy limitations.
    \item \textit{(Hypothetical)} \textbf{Objective} trajectory ($\mathcal{J}^h_{O}$): this is a trajectory where all conflict resolution is optimal and derives from the ground truth extensions. Dialogues do not exist and agents have perfect information at cost 0. It hypothesises what would happen if all agents had full mutual knowledge and always made the right decisions.
\end{itemize}

These trajectories are advantageously intuitive and serve as an example of application-based representations of policies under different perspectives. An illustration of these trajectories can be seen in Figure~\ref{fig:compare-trajec}. If agents have efficient strategies for selecting relevant information/explanations, their Nominal trajectories will converge to the objectively correct answer more frequently (and more quickly), thus becoming increasingly similar to the Objective hypothesis. Likewise, if strategies are efficient, subjective unfairness will be reduced as agents will be able to get to the end of their dialogues -- and thus the decisions made in the Nominal trajectory and its corresponding Subjective hypothesis will agree.

Let $\mathcal{J}_1, \mathcal{J}_2$ denote any 2 trajectories. We calculate dissimilarity measures across comparable trajectories using the discrete Fréchet distance\footnote{\textit{`Informally, it is the minimum length of a leash required to connect a dog, walking along a trajectory $\mathcal{J}_1$, and its owner, walking along a trajectory $\mathcal{J}_2$, as they walk without backtracking along their respective curves from one endpoint to the other.'} - \citet{Agarwal2014ComputingTime}} $Frec(\mathcal{J}_1, \mathcal{J}_2)$. We repeated the tests with two other methods (Partial Curve Mapping and Dynamic Time Warp) and achieved similar results. We chose the aforementioned metric for its intuitive value and simplicity, and will omit the other results as they are redundant. For every strategy, we will compare 3 pairs of trajectories:

\begin{itemize}
    \item \textbf{Objective Unfairness} ($\Omega = Frec(\mathcal{J}_N, \mathcal{J}^h_{O}$)). This represents the level of agreement between the Nominal trajectory and the hypothetical Objective trajectory in a perfect scenario. This mostly captures how objectively correct the decisions were. In Figure~\ref{fig:global-results}, it is possible to observe that \texttt{offensive} and \texttt{defensive} exhibit much more agreement with the correct trajectory than the other two strategies.
    \item \textbf{Subjective Unfairness} ($\Omega_p = Frec(\mathcal{J}_N, \mathcal{J}^h_{S}$)). This represents the level of agreement between the Nominal trajectory and \textit{what the agent considers to have been the right action}, i.e. the hypothetical Subjective trajectory. High levels of disagreement in are representative of high global subjective unfairness - as many agents have very different perceptions of what the right trajectories were. If privacy budgets are minimal and communication is near-impossible, most subjective trajectories will be straight lines, as no agent will ever be fully convinced to alter their route.
    \item \textbf{Subjectivity Gap} ($l_{S} = Frec(\mathcal{J}_N, \mathcal{J}^h_{S})$). This is an assessment on how accurate is the agents' \textit{perception of correctness} to the \textit{actual correctness}. This is measured as the agreement between the hypothetical Subjective and Objective trajectories. We name this phenomenon as the `subjectivity gap.'
\end{itemize}

\begin{figure}
\includegraphics[width=\linewidth]{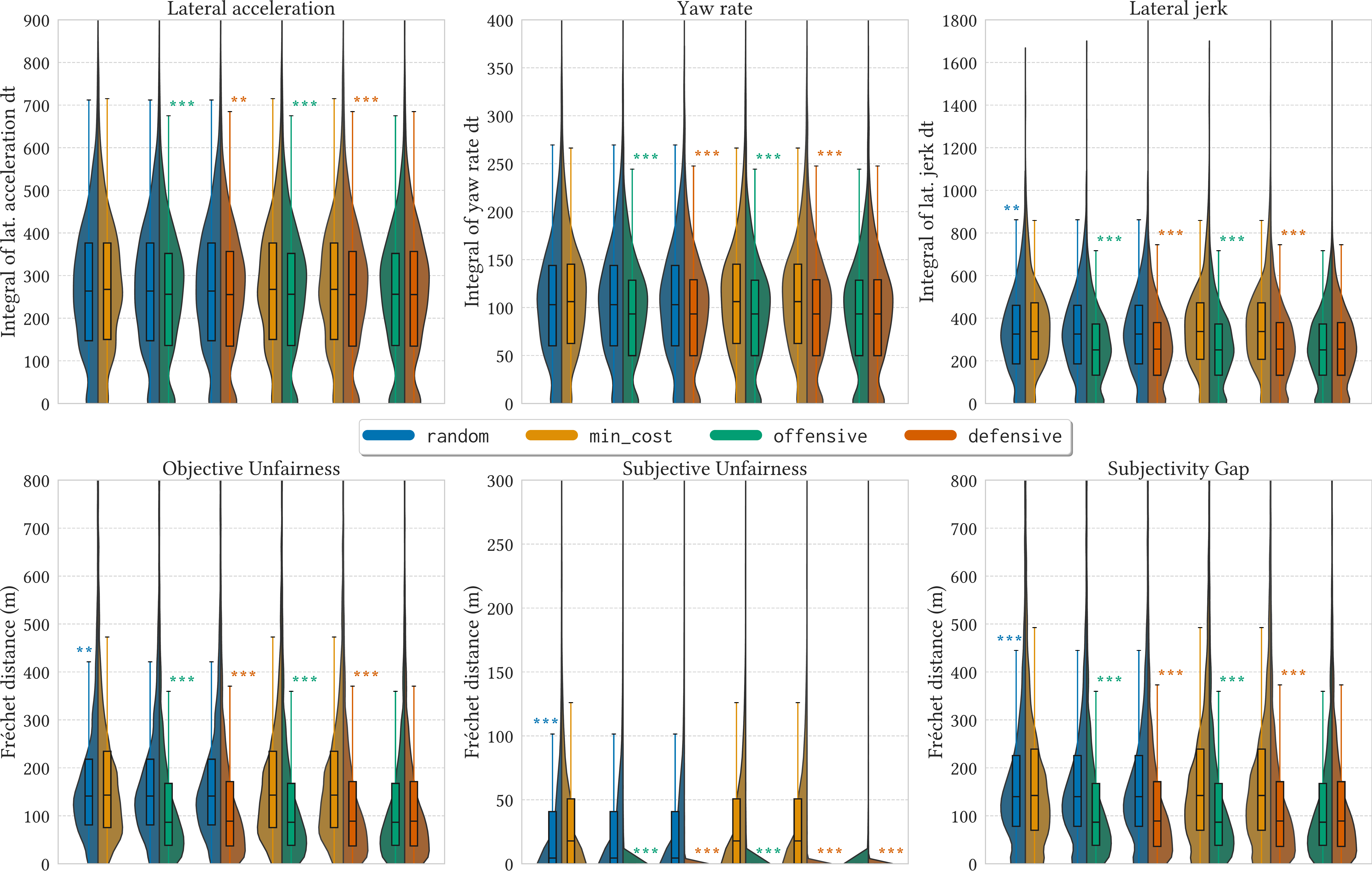}
\caption{Aggregated data of 100 full simulations with a fixed combined privacy budget $2g=60$. Each strategy is shown in a different colour. We perform pairwise statistical two-sided and one-sided Student's t-tests between distributions to check if they differ significantly, and if so, which one presents lower values ((*) $p < 0.05$; (**) $p < 0.01$; (***) $p < 0.001$). The top row shows an aggregation of all agents' simulation metrics ($n = 1600$ per strategy), where lower is indicative of a more comfortable ride. Each data point in the top plots is the area under the curve for that metric (obtained via integration). The bottom plots show values of objective unfairness, subjective unfairness and the subjectivity gap, measured as the Fréchet distance of the corresponding trajectories ($n=32000$ per strategy). Lower values are indicative of more agreement between trajectories.}
\label{fig:global-results}
\end{figure}

The results in the bottom row of Figure~\ref{fig:global-results} show aggregated results for objective and subjective unfairness, as well as the subjectivity gap for each strategy. We perform pairwise comparisons to evaluate how different fare against each other in terms of higher ride comfort (lower dynamics) and lower unfairness values. Higher values in the first row indicate that, when using that strategy, agents experienced higher motion dynamics associated with discomfort. Likewise, higher values in the second row demonstrate a higher disagreement between the pairs of trajectories when agents use that strategy. We observe that \texttt{min\_cost} exhibits the worst performance in all metrics against all other strategies. In second-to-last, \texttt{random} is bested by the other 2 strategies. More remarkably, the winning strategies \texttt{offensive} and \texttt{defensive} exhibit mostly equivalent behaviour, and both are able to practically eradicate occurrences of subjective unfairness within that given privacy budget.

\section{Discussion}

Our experimental results explored four intuitive strategies for conflict resolution under privacy constraints: \texttt{random}, which does not exploit any problem structure; \texttt{min\_cost}, which greedily exploits the cost structure; and \texttt{offensive} and \texttt{defensive}, which are informed by the structure of the underlying culture. The strategy performance follows an expected order. Any informed strategy outperforms the baseline \texttt{random}. The exploitation of the culture structure in \texttt{offensive} and \texttt{defensive} yields better results than just using privacy information, as in \texttt{min\_cost}. Finally, \texttt{defensive} is expected to be the best strategy, since it follows a more natural principle in argumentation: that an argument which has no attackers must be accepted. In this scenario, choosing arguments with fewer attackers makes it more likely that this argument is final. Although many others (more involved) strategies could be used, our goal is to measure the effect of using informed strategies as opposed to strategies unaffected by the argumentative structure (the culture). That is, whether a better usage of the argumentative information leads to better results in both subjective and objective fairness, and other performance metrics in the experimental scenario, under identical privacy constraints.

\subsection{Randomised Cultures}

In the random experiment, Figure~\ref{fig:local_global} shows that \texttt{defensive} yields lower losses both locally and globally, especially for realistic privacy budgets that do not overly limit the length of dialogue while meaningfully constraining the amount of shared information. When unrestricting dialogues, \texttt{defensive} was also able to close dialogues with lower privacy costs. As expected, our results suggest that choosing better strategies imply Pareto dominance. 

When the privacy budget is near zero, all strategies perform similarly poorly. Little dialogue can occur, so agents lack the information to make objectively good decisions and to provide meaningful justifications. With high privacy budgets, strategies perform similarly well. Predictably, all strategies converge to zero subjective local fairness because dialogues can extend indefinitely. The distinct non-zero convergence values of the objective global fairness loss is due primarily to the ground truth used in these experiments.

Our ground truth corresponds to the sceptical acceptance of the proponent's motion. The motion is only validated if \textit{all} attacking arguments are defeated. Having at least one reason to defeat the motion is sufficient to preserve the status quo.\footnote{Since features are randomly uniformly sampled, the likelihood of a victory is associated with the probability of an agent prevailing in multiple random challenges.}  If both agents win in some unattacked arguments and lose in others, the outcome will always lean towards the opponent, even when they swap roles.

\subsection{Multi-Agent Experiment}

We observed similar results in the applied multi-agent experiment. Although the distinction between \texttt{defensive} and \texttt{offensive} was not as clear as in with randomised cultures, both still exhibited superior results and reinforced the argument towards the importance of selecting information by relevance. The formalism proposed in Section~\ref{section:fairness} allows for representing a range of different problems, endowing decentralised systems with considerations of subjectivity that can be modelled for non-human agents. 

By comparing trajectories under different perspectives, one could be tempted to see the Objective hypothesis through the lens of a more `traditional' aspect of trajectories, namely, by how much \textit{shorter} and \textit{quicker} they are, or by observing the system behaviour with metrics such as makespan or flowtime. In our case, better trajectories are not necessarily `optimal' in the traditional sense. They do not attempt to be shorter or quicker, but instead consider ride comfort metrics and ultimately select for better agreement between what actually happens, what \textit{actually} should have happened, and what the agents \textit{believe} should have happened.

The environment with 16 agents avoiding each other through combinations of artificial potential fields in continuous space with simulated physics is dense and complicated, and can lead to edge cases where agents' decisions can cascade and propagate to multiple others (especially if radii are large). Notwithstanding all the potential for results being confounded by this property of the system, we still managed to demonstrate statistically significant results for superior global performance with better strategies. We expect even clearer results in discrete applications, topological representations \citep{Bhattacharya2012TopologicalPlanning}, or in environments with planning in mind, such as Conflict-Based Search methods \citep{Sharon2015Conflict-basedPathfinding}.

For real applications, the properties and structure of the chosen alteroceptive culture play an important role in indicating which strategy will perform better towards the aimed Pareto dominance within the fairness-privacy trade-off. Perhaps with the exception of degenerate cases of cultures, our heuristics should provide a good guide for considering the choice of arguments for building explanations in fairness-aware conflict resolution disputes, especially if one cannot assess the relevance of the \textit{content} of the arguments (what it actually means), but has access to its structure (the ruleset that originated the culture, and the relationships between rules).

\section{Conclusion}

In this work, we propose \textit{perspective} and \textit{scope} as new considerations for the problem of fairness in decentralised conflict resolution. We show how privacy limitations introduce partial observability, and consequently, a \textit{trade-off} for fairness losses. Our proposed architecture for privacy-aware conflict resolution allows for the representation and resolution of such conflicts in multi-agent settings, and underpins our experimental setup.

Our central insight shows that simulating agents' actions from the exclusive perspective of their local knowledge and comparing them to the \textit{actual} executed behaviour can grant an understanding of their \textit{subjective} perceptions of fairness. Moreover, comparing this subjective perception to an omniscient \textit{objective} behaviour can also provide an insight on how well-informed agents are in their perception of the world. Different environments and applications could yield interesting combinations of these properties. For instance, a multi-agent society with low global subjective unfairness and a high subjectivity gap could indicate that agents are not acting in accordance to the desired notion of fairness, either by ignorance, design flaw, or adversarial intent.

This work presents no shortage of avenues for future studies. To name a few, extending the scope of interactions beyond pairwise conflict resolution can be investigated by means of collective argumentation \citep{Bodanza2017CollectiveFrameworks}. One also could observe the effects of the fairness-privacy trade-off in populations with heterogeneous strategies, and, in like manner, the consequences of heterogeneous privacy budget distributions for subjective and objective fairness, or even extending the dialogues to more general dispute trees, allowing agents to backtrack and find new lines of defence. To a wider regard, the use of cultures as mechanisms for explainable conflict resolution in rule-aware environments encourages similar reflection for our work. Our findings further reinforce the importance of (good) explanations and explanatory strategies in multi-agent systems, this time as mitigatory instruments for applications where subjectivity is a concern.

Inescapably, agents `agree to disagree' when compelled to concede a resource due to ulterior reasons, such as privacy preservation or time constraints. Whenever continuing the dispute would \textit{produce effects more undesirable than the loss of the dispute} in itself, agents tolerate those consequences and will prefer to abandon the dispute. This phenomenon is not necessarily new for humans, but extending this consideration for societies of artificial agents grants a dimension of subjectivity that is worth exploring in multi-agent systems research. Preparing non-human agents and systems to \textit{consider} subjective unfairness can be an important step in facilitating the integration of human agents in future hybrid multi-agent societies. We encourage readers and researchers in the field to regard subjectivity as an important consideration in fairness-aware decentralised systems with incomplete information.

\section*{Acknowledgements}
Alex Raymond was supported by the Royal Commission for the Exhibition of 1851 and L3Harris ASV. Matthew Malencia was supported by the National Science Foundation Graduate Research Fellowship under Grant DGE-1845298. Guilherme Paulino-Passos was supported by CAPES (Brazil, Ph.D. Scholarship 88881.174481/2018-01). A. Prorok was supported by the Engineering and Physical Sciences Research Council (grant EP/S015493/1). Their support is gratefully acknowledged.

\bibliographystyle{ACM-Reference-Format}  %

\begin{thebibliography}{41}
\providecommand{\natexlab}[1]{#1}
\expandafter\ifx\csname urlstyle\endcsname\relax
  \providecommand{\doi}[1]{doi:\discretionary{}{}{}#1}\else
  \providecommand{\doi}{doi:\discretionary{}{}{}\begingroup
  \urlstyle{rm}\Url}\fi
\providecommand{\selectlanguage}[1]{\relax}
\providecommand{\bibAnnoteFile}[1]{%
  \IfFileExists{#1}{\begin{quotation}\noindent\textsc{Key:} #1\\
  \textsc{Annotation:}\ \input{#1}\end{quotation}}{}}
\providecommand{\bibAnnote}[2]{%
  \begin{quotation}\noindent\textsc{Key:} #1\\
  \textsc{Annotation:}\ #2\end{quotation}}

\bibitem[{Agarwal et~al.(2014)Agarwal, Avraham, Kaplan, and
  Sharir}]{Agarwal2014ComputingTime}
Agarwal, P.~K., Avraham, R.~B., Kaplan, H., and Sharir, M. (2014).
\newblock {Computing the Discrete Fr{\'{e}}chet Distance in Subquadratic Time}.
\newblock \emph{SIAM Journal on Computing} 43, 429--449.
\newblock \doi{10.1137/130920526}
\input{Agarwal2014ComputingTime}

\bibitem[{Amgoud and Prade(2009)}]{Amgoud2009UsingDecisions}
Amgoud, L. and Prade, H. (2009).
\newblock {Using arguments for making and explaining decisions}.
\newblock \emph{Artificial Intelligence} 173, 413--436
\input{Amgoud2009UsingDecisions}

\bibitem[{Bertsimas et~al.(2013)Bertsimas, Farias, and
  Trichakis}]{bertsimas_fairness_2013}
Bertsimas, D., Farias, V.~F., and Trichakis, N. (2013).
\newblock Fairness, efficiency, and flexibility in organ allocation for kidney
  transplantation.
\newblock \emph{Operations Research} 61, 73--87.
\newblock Publisher: INFORMS
\input{bertsimas_fairness_2013}

\bibitem[{Bhattacharya et~al.(2012)Bhattacharya, Likhachev, and
  Kumar}]{Bhattacharya2012TopologicalPlanning}
Bhattacharya, S., Likhachev, M., and Kumar, V. (2012).
\newblock {Topological constraints in search-based robot path planning}.
\newblock \emph{Autonomous Robots} 33, 273--290.
\newblock \doi{10.1007/s10514-012-9304-1}
\input{Bhattacharya2012TopologicalPlanning}

\bibitem[{Bin-Obaid and Trafalis(2018)}]{bin-obaid_fairness_2018}
Bin-Obaid, H.~S. and Trafalis, T.~B. (2018).
\newblock Fairness in {Resource} {Allocation}: {Foundation} and {Applications}.
\newblock In \emph{International {Conference} on {Network} {Analysis}}
  (Springer), 3--18
\input{bin-obaid_fairness_2018}

\bibitem[{Binmore(1992)}]{Binmore1992FunGames}
Binmore, K. (1992).
\newblock {Fun and games}.
\newblock \emph{A text on game theory}
\input{Binmore1992FunGames}

\bibitem[{Blair et~al.(1993)Blair, Mutchler, and
  Liu}]{Blair1993GamesInformation}
Blair, J. R.~S., Mutchler, D., and Liu, C. (1993).
\newblock {Games with imperfect information}.
\newblock In \emph{Proceedings of the AAAI Fall Symposium on Games: Planning
  and Learning, AAAI Press Technical Report FS93-02, Menlo Park CA}. 59--67
\input{Blair1993GamesInformation}

\bibitem[{Bodanza et~al.(2017)Bodanza, Tohm{\'{e}}, and
  Auday}]{Bodanza2017CollectiveFrameworks}
Bodanza, G., Tohm{\'{e}}, F., and Auday, M. (2017).
\newblock {Collective argumentation: A survey of aggregation issues around
  argumentation frameworks}.
\newblock \emph{Argument {\&} Computation} 8, 1--34.
\newblock \doi{10.3233/AAC-160014}
\input{Bodanza2017CollectiveFrameworks}

\bibitem[{{\v{C}}yras et~al.(2019){\v{C}}yras, Birch, Guo, Toni, Dulay, Turvey
  et~al.}]{Cyras2019ExplanationsDispute}
{\v{C}}yras, K., Birch, D., Guo, Y., Toni, F., Dulay, R., Turvey, S., et~al.
  (2019).
\newblock {Explanations by arbitrated argumentative dispute}.
\newblock \emph{Expert Systems with Applications} 127, 141--156.
\newblock \doi{https://doi.org/10.1016/j.eswa.2019.03.012}
\input{Cyras2019ExplanationsDispute}

\bibitem[{Doutre and Mengin(2001)}]{Doutre2001PreferredComputation}
Doutre, S. and Mengin, J. (2001).
\newblock {Preferred extensions of argumentation frameworks: Query, answering,
  and computation}.
\newblock In \emph{International Joint Conference on Automated Reasoning}.
  272--288
\input{Doutre2001PreferredComputation}

\bibitem[{Dung(1995)}]{Dung1995OnGames}
Dung, P.~M. (1995).
\newblock {On the acceptability of arguments and its fundamental role in
  nonmonotonic reasoning, logic programming and n-person games}.
\newblock \emph{Artificial Intelligence} 77, 321--357
\input{Dung1995OnGames}

\bibitem[{Dwork et~al.(2014)Dwork, Roth, and {others}}]{dwork_algorithmic_2014}
Dwork, C., Roth, A., and {others} (2014).
\newblock The algorithmic foundations of differential privacy.
\newblock \emph{Foundations and Trends® in Theoretical Computer Science} 9,
  211--407.
\newblock Publisher: Now Publishers, Inc.
\input{dwork_algorithmic_2014}

\bibitem[{Emelianov et~al.(2019)Emelianov, Arvanitakis, Gast, Gummadi, and
  Loiseau}]{emelianov_price_2019}
Emelianov, V., Arvanitakis, G., Gast, N., Gummadi, K., and Loiseau, P. (2019).
\newblock The price of local fairness in multistage selection.
\newblock \emph{arXiv preprint arXiv:1906.06613}
\input{emelianov_price_2019}

\bibitem[{Fan and Toni(2015)}]{Fan2015OnArgumentation}
Fan, X. and Toni, F. (2015).
\newblock {On Computing Explanations in Argumentation}.
\newblock \emph{Twenty-Ninth AAAI Conference on Artificial Intelligence}
\input{Fan2015OnArgumentation}

\bibitem[{Felv{\'{e}}gi(2013)}]{Felvegi2013ShipsEffect}
[Dataset] Felv{\'{e}}gi, C. (2013).
\newblock {Ships with ripple effect}
\input{Felvegi2013ShipsEffect}

\bibitem[{Gao et~al.(2016)Gao, Toni, Wang, and
  Xu}]{Gao2016Argumentation-basedPreserved}
Gao, Y., Toni, F., Wang, H., and Xu, F. (2016).
\newblock {Argumentation-based multi-agent decision making with privacy
  preserved}.
\newblock In \emph{Proceedings of the 2016 International Conference on
  Autonomous Agents {\&} Multiagent Systems}. 1153--1161
\input{Gao2016Argumentation-basedPreserved}

\bibitem[{Jakobovits and Vermeir(1999)}]{Jakobovits1999DialecticFrameworks}
Jakobovits, H. and Vermeir, D. (1999).
\newblock {Dialectic semantics for argumentation frameworks}.
\newblock In \emph{Proceedings of the seventh international conference on
  Artificial intelligence and law - ICAIL '99} (New York, New York, USA: ACM
  Press), 53--62
\input{Jakobovits1999DialecticFrameworks}

\bibitem[{Li and Tracer(2017)}]{li_interdisciplinary_2017}
Li, M. and Tracer, D.~P. (2017).
\newblock \emph{Interdisciplinary {Perspectives} on {Fairness}, {Equity}, and
  {Justice}} (Springer)
\input{li_interdisciplinary_2017}

\bibitem[{Linkovich(2016)}]{Linkovich2016Carphysics2d}
[Dataset] Linkovich, M. (2016).
\newblock {carphysics2d}
\input{Linkovich2016Carphysics2d}

\bibitem[{Malmi et~al.(2015)Malmi, Tatti, and Gionis}]{Malmi2015BeyondGraphs}
Malmi, E., Tatti, N., and Gionis, A. (2015).
\newblock {Beyond rankings: comparing directed acyclic graphs}.
\newblock \emph{Data mining and knowledge discovery} 29, 1233--1257
\input{Malmi2015BeyondGraphs}

\bibitem[{Marx(1875)}]{Marx1875CritiqueGotha}
Marx, K. (1875).
\newblock \emph{{Critique of the Social Democratic Program of Gotha}} (Letter
  to Bracke)
\input{Marx1875CritiqueGotha}

\bibitem[{Modgil and Luck(2009)}]{Modgil2009ArgumentationGoals}
Modgil, S. and Luck, M. (2009).
\newblock {Argumentation Based Resolution of Conflicts between Desires and
  Normative Goals}.
\newblock In \emph{Argumentation in Multi-Agent Systems} (Springer, Berlin,
  Heidelberg). 19--36.
\newblock \doi{10.1007/978-3-642-00207-6{\_}2}
\input{Modgil2009ArgumentationGoals}

\bibitem[{Monster(2003)}]{Monster2003CarGames}
[Dataset] Monster, M. (2003).
\newblock {Car physics for games}
\input{Monster2003CarGames}

\bibitem[{Moskop and Iserson(2007)}]{Moskop2007TriagePrinciples}
Moskop, J.~C. and Iserson, K.~V. (2007).
\newblock {Triage in Medicine, Part II: Underlying Values and Principles}.
\newblock \emph{Annals of Emergency Medicine} 49, 282--287.
\newblock \doi{10.1016/j.annemergmed.2006.07.012}
\input{Moskop2007TriagePrinciples}

\bibitem[{Narayanan(2018)}]{narayanan_translation_2018}
Narayanan, A. (2018).
\newblock Translation tutorial: 21 fairness definitions and their politics.
\newblock In \emph{Proc. {Conf}. {Fairness} {Accountability} {Transp}., {New}
  {York}, {USA}}
\input{narayanan_translation_2018}

\bibitem[{Nguyen and Rothe(2016)}]{nguyen_local_2016}
Nguyen, N.-T. and Rothe, J. (2016).
\newblock Local fairness in hedonic games via individual threshold coalitions.
\newblock In \emph{Proceedings of the 2016 {International} {Conference} on
  {Autonomous} {Agents} \& {Multiagent} {Systems}}. 232--241
\input{nguyen_local_2016}

\bibitem[{Nguyen et~al.(2019)Nguyen, NguyenDinh, Lechner, and
  Wong}]{Nguyen2019InsightSingapore}
Nguyen, T., NguyenDinh, N., Lechner, B., and Wong, Y.~D. (2019).
\newblock {Insight into the lateral ride discomfort thresholds of young-adult
  bus passengers at multiple postures: Case of Singapore}.
\newblock \emph{Case Studies on Transport Policy} 7, 617--627.
\newblock \doi{https://doi.org/10.1016/j.cstp.2019.07.002}
\input{Nguyen2019InsightSingapore}

\bibitem[{Niskanen and J{\"{a}}rvisalo(2020)}]{Niskanen2020-toksia:Reasoner}
Niskanen, A. and J{\"{a}}rvisalo, M. (2020).
\newblock {\selectlanguage{English}{{$\mu$}-toksia: An Efficient Abstract
  Argumentation Reasoner}}.
\newblock In \emph{Proceedings of the 17th International Conference on
  Principles of Knowledge Representation and Reasoning (KR 2020)} (United
  States: AAAI Press)
\input{Niskanen2020-toksia:Reasoner}

\bibitem[{Prorok and Kumar(2017)}]{Prorok2017Privacy-preservingSystems}
Prorok, A. and Kumar, V. (2017).
\newblock {Privacy-preserving vehicle assignment for mobility-on-demand
  systems}.
\newblock In \emph{2017 IEEE/RSJ International Conference on Intelligent Robots
  and Systems (IROS)}. 1869--1876
\input{Prorok2017Privacy-preservingSystems}

\bibitem[{Rawls(1991)}]{Rawls1991JusticeMetaphysical}
Rawls, J. (1991).
\newblock {Justice as fairness: Political not metaphysical}.
\newblock In \emph{Equality and Liberty} (Springer). 145--173
\input{Rawls1991JusticeMetaphysical}

\bibitem[{Raymond et~al.(2020)Raymond, Gunes, and
  Prorok}]{Raymond2020Culture-BasedDeconfliction}
Raymond, A., Gunes, H., and Prorok, A. (2020).
\newblock {Culture-Based Explainable Human-Agent Deconfliction}.
\newblock In \emph{Proceedings of the 19th International Conference on
  Autonomous Agents and MultiAgent Systems} (Richland, SC: International
  Foundation for Autonomous Agents and Multiagent Systems), AAMAS ’20,
  1107--1115
\input{Raymond2020Culture-BasedDeconfliction}

\bibitem[{Reif(1984)}]{Reif1984TheInformation}
Reif, J.~H. (1984).
\newblock {The complexity of two-player games of incomplete information}.
\newblock \emph{Journal of Computer and System Sciences} 29, 274--301.
\newblock \doi{https://doi.org/10.1016/0022-0000(84)90034-5}
\input{Reif1984TheInformation}

\bibitem[{Rosenfeld and Richardson(2019)}]{Rosenfeld2019ExplainabilitySystems}
Rosenfeld, A. and Richardson, A. (2019).
\newblock {Explainability in human–agent systems}.
\newblock \emph{Autonomous Agents and Multi-Agent Systems} 33, 673--705
\input{Rosenfeld2019ExplainabilitySystems}

\bibitem[{Selbst et~al.(2019)Selbst, Boyd, Friedler, Venkatasubramanian, and
  Vertesi}]{selbst_fairness_2019}
Selbst, A.~D., Boyd, D., Friedler, S.~A., Venkatasubramanian, S., and Vertesi,
  J. (2019).
\newblock Fairness and {Abstraction} in {Sociotechnical} {Systems}.
\newblock In \emph{Proceedings of the {Conference} on {Fairness},
  {Accountability}, and {Transparency}} (New York, NY, USA: Association for
  Computing Machinery), {FAT}* ’19, 59--68.
\newblock \doi{10.1145/3287560.3287598}.
\newblock Event-place: Atlanta, GA, USA
\input{selbst_fairness_2019}

\bibitem[{Sharon et~al.(2015)Sharon, Stern, Felner, and
  Sturtevant}]{Sharon2015Conflict-basedPathfinding}
Sharon, G., Stern, R., Felner, A., and Sturtevant, N.~R. (2015).
\newblock {Conflict-based search for optimal multi-agent pathfinding}.
\newblock \emph{Artificial Intelligence} 219, 40--66.
\newblock \doi{10.1016/j.artint.2014.11.006}
\input{Sharon2015Conflict-basedPathfinding}

\bibitem[{Sovrano and Vitali(2021)}]{Sovrano2021FromExplanation}
Sovrano, F. and Vitali, F. (2021).
\newblock {From Philosophy to Interfaces: An Explanatory Method and a Tool
  Inspired by Achinstein’s Theory of Explanation}.
\newblock In \emph{26th International Conference on Intelligent User
  Interfaces} (New York, NY, USA: Association for Computing Machinery), IUI
  '21, 81--91.
\newblock \doi{10.1145/3397481.3450655}
\input{Sovrano2021FromExplanation}

\bibitem[{Such et~al.(2014)Such, Espinosa, and
  Garcia-Fornes}]{Such2014ASystems}
Such, J.~M., Espinosa, A., and Garcia-Fornes, A. (2014).
\newblock {A survey of privacy in multi-agent systems}.
\newblock \emph{The Knowledge Engineering Review} 29, 314--344.
\newblock \doi{10.1017/S0269888913000180}
\input{Such2014ASystems}

\bibitem[{{Sunseeker}(2019)}]{Sunseeker2019SunseekerBrochure}
[Dataset] {Sunseeker} (2019).
\newblock {Sunseeker Hawk 38 Brochure}
\input{Sunseeker2019SunseekerBrochure}

\bibitem[{Torreno et~al.(2017)Torreno, Onaindia, Komenda, and
  {\v{S}}tolba}]{Torreno2017CooperativeSurvey}
Torreno, A., Onaindia, E., Komenda, A., and {\v{S}}tolba, M. (2017).
\newblock {Cooperative multi-agent planning: A survey}.
\newblock \emph{ACM Computing Surveys (CSUR)} 50, 1--32
\input{Torreno2017CooperativeSurvey}

\bibitem[{Verma and Rubin(2018)}]{verma_fairness_2018}
Verma, S. and Rubin, J. (2018).
\newblock Fairness {Definitions} {Explained}.
\newblock In \emph{2018 {IEEE}/{ACM} {International} {Workshop} on {Software}
  {Fairness} ({FairWare})}. 1--7
\input{verma_fairness_2018}

\bibitem[{Warren(1990)}]{Warren1990MultipleFields}
Warren, C.~W. (1990).
\newblock {Multiple robot path coordination using artificial potential fields}.
\newblock In \emph{Proceedings., IEEE International Conference on Robotics and
  Automation}. 500--505.
\newblock \doi{10.1109/ROBOT.1990.126028}
\input{Warren1990MultipleFields}

\end{thebibliography}

\end{document}